\pdfoutput=1
\documentclass[12pt, hidelinks]{article}

\usepackage[twoside=true, margin=1in]{geometry}
\usepackage[numbers]{natbib}
\usepackage{amssymb}

\usepackage{preample}

\title{Credible Decentralized Exchange Design via Verifiable Sequencing Rules}

\author{Matheus V. X. Ferreira\thanks{Harvard University (\href{mailto:matheus@seas.harvard.edu}{matheus@seas.harvard.edu}).}
\and
David C. Parkes\thanks{Harvard University (\href{mailto:parkes@eecs.harvard.edu}{parkes@eecs.harvard.edu}).}
}

\begin{document}

\maketitle

\begin{abstract}
Trading on decentralized exchanges has been one of the primary use cases for permissionless blockchains with daily trading volume exceeding billions of U.S.~dollars. In the status quo, users broadcast transactions they wish to execute in the exchange and miners are responsible for composing a block of transactions and picking an execution ordering --- the order in which transactions execute in the exchange. Due to the lack of a regulatory framework, it is common to observe miners exploiting their privileged position by front-running transactions and obtaining risk-fee profits. Indeed, the Flashbots service institutionalizes this exploit, with miners auctioning the right to front-run transactions. In this work, we propose to modify the interaction between miners and users and initiate the study of {\em verifiable sequencing rules}. As in the status quo, miners can determine the content of a block; however, they commit to respecting a sequencing rule that constrains the execution ordering and is verifiable (there is a polynomial time algorithm that can verify if the execution ordering satisfies such constraints). Thus in the event a miner deviates from the sequencing rule, anyone can generate a proof of non-compliance.

We ask if there are sequencing rules that limit price manipulation from miners in a two-token liquidity pool exchange. Our first result is an impossibility theorem: for any sequencing rule, there is an instance of user transactions where the miner can obtain non-zero risk-free profits. In light of this impossibility result, our main result is a verifiable sequencing rule that provides execution price guarantees for users. In particular, for any user transaction $A$, it ensures that either (1) the execution price of $A$ is at least as good as if $A$ was the only transaction in the block, or (2) the execution price of $A$ is worse than this ``standalone'' price and the miner does not gain when including $A$ in the block. Our framework does not require users to use countermeasures against predatory trading strategies, for example, set limit prices or split large transactions into smaller ones. This is likely to improve user experience relative to the status quo.
\end{abstract}


\section{Introduction}\label{sec:introduction}

Decentralized finance,  also referred to as {\em DeFi}, has been one of the main applications of permissionless blockchains. DeFi protocols allow liquidity providers to lock capital into smart contracts. The locked money enables the liquidity provider to obtain revenue from transaction fees by providing liquidity for financial services such as lending and trading. As of 2022, the locked capital in DeFi protocols exceeds $\$40$ billion U.S.~dollars~\cite{defipulse}, with the most prominent decentralized exchange, Uniswap~\cite{adams2021uniswap}, having over $\$7$ billion U.S.~dollars in reserves and trading volume that often exceeds billions of U.S.~dollars per day~\cite{uniswapvolume}.

In theory, a decentralized exchange allows traders to submit buy or sell orders to a smart contract without needing any intermediary. Uniswap implements a liquidity pool decentralized exchange as follows. Liquidity providers {\em lock} capital into a liquidity pool of two token types. Let $X_i$ be the amount of token $i \in \{1, 2\}$ locked on a liquidity pool. The product $X_1 \cdot X_2$ defines the {\em potential} corresponding to the current state of the exchange. A user can submit a trade on Uniswap against the liquidity pool and, for example, withdraw $q>0$ units of token $1$ as long as they deposit a quantity $p>0$ of token $2$ that preserves the potential:
\begin{equation}\label{eq:uniswap}
(X_1 - q) \cdot (X_2 + p) = X_1 \cdot X_2.
\end{equation}

Many alternatives exist to the product potential function of Uniswap~\cite{egorov2019stableswap, martinelli2019non}, and Uniswap itself has been generalized in recent years to Uniswap v3, which allows for the amount of locked liquidity to vary by price. Still, a common feature of {\em liquidity pool exchanges} is the existence of a reserve $X \in \mathbb R^2$ (or, in general, $X \in \mathbb R^n$ for $n$ tokens) representing the locked capital (and the exchange state). Users submit orders to buy or sell a particular token, modifying the exchange state after execution. In the example~\eqref{eq:uniswap}, the exchange state changes from $(X_1, X_2)$ to $(X_1 - q, X_2 + p)$ when the user trades $p$ units of token $2$ for $q$ units of token $1$.

In the blockchain setting, where computation and storage resources are highly scarce, liquidity pool exchanges provide benefits compared with traditional order books. First, the memory storage is constant while an order book's memory requirement grows with the number of pending transactions. In terms of computation, executing an order requires only a constant number of operations, while an order book requires matching buys with sells and updating underlying data structures. Regarding latency, orders execute instantaneously when processed, while order books need each buy/sell order to wait to match with a corresponding sell/buy order.

In an ideal setting, users would privately submit orders to the liquidity pool exchange; in this sense, there would be no intermediary. In practice, miners act as intermediaries between users and the exchange. Miners choose which pending orders to include in a block and, in the status quo,  are also responsible for specifying the order with which transactions execute, which we refer to as the {\em execution ordering}.\footnote{We refer to a miner as the agent that defines the content and ordering of transactions inside a block. This is the most common terminology for Proof-of-Work blockchains like Bitcoin. On the other hand, Proof-of-Stake blockchains like Ethereum refer to these agents as block proposers. The distinction between a Proof-of-Work miner and a Proof-of-Stake block proposer is irrelevant to our results since they play the same role. Moreover, decentralized exchanges can be implemented in either Proof-of-Work blockchains or Proof-of-Stake blockchains and are equally vulnerable to front-running by miners/proposers. Indeed, Uniswap was initially implemented over Ethereum Proof-of-Work while their most recent implementation operates over Ethereum Proof-of-Stake.}

As one can observe from~\eqref{eq:uniswap}, the state at which an order executes highly influences the trade outcome. Not surprisingly, miners can take advantage of their role to manipulate the content of a block and its execution ordering~\cite{park2021conceptual}. This is known as {\em Miner Extractable Value} (MEV), also referred to as maximum/maximal extractable value. A well-documented, front-running attack from miners is a {\em sandwich attack}~\cite{daian2020flash, zhou2021high}.\footnote{See \cite{daian2020flash, weintraub2022flash} for a broader discussion and empirical measurements on different kinds of MEV.} Here, a malicious miner purchases $x>0$ units of a token ahead of  $q>0$ units purchased by a user, and then immediately sells the $x$ units. The effect is that the miner achieves risk-free profits at the cost of a higher price for the user. This attack has been institutionalized through the {\em front-running-as-a-service} business model provided by  Flashbots~\cite{flashbots, weintraub2022flash}, where miners auction the right to allow another party to insert orders and manipulate the order execution.

Although front-running schemes are familiar to traditional finance, a history of regulation in the financial system is designed to protect traders from this kind of market manipulation~\cite{brunnermeier2009fundamental}. However, these regulations have not been enforced on decentralized exchanges where anonymous entities can become a miner and create blocks. Given the absence of regulatory enforcement, there is a great interest in developing algorithmic techniques to mitigate market manipulation by miners (or other entities). Batch auctions are one such approach~\cite{budish2015high, gnosis, mcmenamin2022fairtradex}, where transactions are batched and all clear at the same market clearing price. However, blockchains execute transactions sequentially rather than in batches. Thus smart-contract-based batch auctions introduce latency and computational overhead.  Instead, we focus on the well-established (sequential) liquidity pool exchange model, where transactions execute sequentially at potentially different prices, due to their computational efficiency and low latency, and we seek to design sequencing rules that can provably mitigate opportunities for miner manipulation.

A {\em trusted relay service}~\cite{capponi2022evolution, flashbots} like Flashbots Protect (developed by the same organization that operationalizes MEV auctions) is another approach to mitigating manipulation and one that has gained relatively rapid adoption.
A user privately communicates their transaction to a trusted service under the promise that the service will not act on privileged information obtained from that user (by injecting transactions and manipulating execution ordering).
The relay service recruits a trusted set of miners to include that transaction in a block.\footnote{For simplicity, we consider a game between users and miners that already captures the challenges of designing a trustworthy decentralized exchange. In practice, the interactions between users with decentralized exchanges implemented over the Ethereum blockchain are becoming increasingly more complex. Since the Ethereum Proof-of-Stake update, most of the interactions between users and the blockchain proceeds as follows. A user sends a transaction to a block builder (like Flashbots protect) that generates a block --- an ordered list of transactions. When constructing the block, the builder includes private transactions (transactions users send directly to Flashbots) and public transactions (transactions users broadcast to the whole network). Flashbots protect allows sandwich attacks on public transactions but promises not to front-run private transactions. The block builder forwards the block to a relay service (like Flashbots Relay) together with a bid. The relay service forwards a hash of the block together with the bid to validators. The validator chooses which block to confirm via an auction: picking the hash with the highest bid (without seeing the block content). Only after the validator commits to a winning hash, the relay reveals the block associated with that hash. This way validators cannot front-run user transactions since they already commit to the block content. Our results also apply to this setting since the status quo requires users to trust that both Flashbots Protect and the Flashbots Relay will not act on privileged information and sandwich attack user transactions.}
Such approaches, however, lack credibility because existing solutions provide no mechanism for which a user can verify that the service manager (or the trusted miner) does not manipulate the execution ordering for profit.

Our approach is similar to a relay service, where users communicate their transactions to the service, and the service sequences these transactions into blocks. However, our approach makes such a relay service  credible: we give a concrete procedure by which any user can verify a set of formal guarantees. To be concrete, we model a miner as the entity that is responsible for picking which transactions to include in a block $B$, and that is free to manipulate the block's content and include its own transactions. However, the miner can commit to picking an execution ordering from a {\em verifiable sequencing rule}. Formally, a {\em sequencing rule} is a function $S$ that takes the initial state $X_0$ (of liquidity reserves), before any transaction executes, the block $B$ (the transactions to include), and outputs a non-empty set  $S(X_0, B)$, which is a set of permutations of $B$. We refer to the elements of $S(X_0, B)$ as {\em valid execution orderings}. We ask that a sequencing rule is {\em efficient}, in the sense that there is a polynomial time algorithm that can compute some $T \in S(X_0, B)$ for any $(X_0, B)$. Moreover, we ask that a sequencing rule is {\em verifiable}, in the sense that anyone can efficiently check if $T \in S(X_0, B)$ for any $(T, X_0, B)$. Hence, any  detectable deviation by a miner from $S$ (picking $T \not\in S(X_0, B)$) can be punished, either financially or via reputation loss.

We ask if there are sequencing rules where the miner cannot profit from manipulating  the block content while respecting the sequencing rule. We summarize our findings as follows:
\begin{enumerate}
    \item {\bf Theorem~\ref{thm:pump-dump} (Folklore).} For a large class of liquidity pool exchanges, including the product potential design of Uniswap, and where miners can pick the content of a block (including adding their own transactions) and sequence transactions as they like, a user might receive an arbitrarily bad execution price. That is, a user who buys $q$ units of a token might make an arbitrarily large payment. Equivalently, a user who sells $q$ units of a token might receive an arbitrarily small payment.
    \item {\bf Theorem~\ref{thm:impossibility}.} For a class of liquidity pool exchanges (that includes Uniswap), for any sequencing rule, there are instances where the miner has a profitable risk-free undetectable deviation.
    \item {\bf Theorem~\ref{thm:greedy}.} We specify a sequencing rule (the Greedy Sequencing Rule) such that, for any valid execution ordering, then for any user transaction $A$ that the miner includes in the block, it must be that either (1) the user efficiently detects the miner did not respect the sequencing rule, or (2) the execution price of $A$ for the user is at least as good as if $A$ was the only transaction in the block, or (3) the execution price of $A$ is worse than this standalone price but the miner does not gain when including $A$ in the block.
\end{enumerate}

Theorem~\ref{thm:pump-dump} shows a miner can force an arbitrarily bad execution (and profit as a result) if they can pick an execution ordering (in the absence of trading costs for the miner). This result can be weakened by introducing transaction fees (paid by each transaction in the block) or trading fees (proportional to the trading volume). However, sufficiently large transactions can still be the victim of a sandwich attack. While prior work has explored (often complex) trade-offs between fees and order size~\cite{heimbach2022eliminating} to mitigate market manipulation, our positive result holds even in the absence of trading costs (blockchain transaction fees or exchange fees).

One can also weaken Theorem~\ref{thm:pump-dump} by allowing users to make use of a {\em limit price} to specify the maximum they want to pay for a buy order (or the minimum they want to receive for a sell order). By setting a limit price, a transaction fails to execute unless their execution price is as good as this limit. If a transaction fails, no trade takes place, but users must still pay blockchain transaction fees. Requiring all transactions to pay fees, even if they fail, is a security mechanism from existing blockchains to increase the cost for a Denial-of-Service-Attack (where a malicious user nationally creates transactions that will fail and consume network resources). 

One might wonder why can't users simply set limit prices equal to the most recent market price to prevent sandwich attacks? First, users can still be manipulated in the status quo such that they trade at their limit prices, and thus, they need to set limit prices close to the standalone price, and this can lead to high latency, with transactions failing to execute. In contrast, our positive result, Theorem~\ref{thm:greedy}, does not require users to use limit prices for protection (since if a transaction is a victim of a predatory trading strategy it still trades at a price at least as good as the standalone price).

If one aims to design a sequencing rule where the miner can  never obtain risk-free profits (risk-free profits meaning the miner is sure to receive some tokens for free), then Theorem~\ref{thm:impossibility} shows such a goal is unattainable. Thus, Theorem~\ref{thm:greedy} focuses on providing provable guarantees from the user's perspective. This is our main result, and ensures that if a self-interested miner includes a user's transaction in the block, then either the transaction executes with at a good execution price---as good as if the user's transaction was the only one in the block---or the miner does not gain by including the transaction. That is, a miner can profitably insert their own transactions, but only to the extent that the user's execution price is no worse than their standalone price (i.e., the price if they were the only transaction in the block). Although a user can still get a bad execution price, but in this case the miner provably does not profit from including the user's transaction. For example, if two users each wish to buy $q$ units of the same token, then in the absence of any other transactions, it is inevitable that the transaction that executes in second place pays a higher price. This is due to competition for the same token and not due to miner manipulation.

\subsection{Technical overview}

During a sandwich attack, a miner manipulates the state of the exchange in a way that causes one or more user transactions to achieve a worse execution price. We formalize the properties achieved by our sequencing rule by taking the price at the most recent state  of the liquidity reserves, $X_0 \in \mathbb R^2$, 
when a user submits a buy or sell order,  as a benchmark. This is a relevant benchmark because the blockchain consensus ensures that $X_0$ is not manipulable by the miner.\footnote{The miner could manipulate $X_0$ over multiple blocks, but we assume a different miner creates each block, or equivalently that miners are myopic. This assumption is well-motivated for a decentralized blockchain where the miner for a block is sampled from a large population and the miner selection mechanism is unpredictable so that the miner for a round is unknown until that round starts. This is a common assumption in the context of transaction fee mechanisms~\cite{ferreira2021dynamic, roughgarden2021transaction, lavi2022redesigning}. Note that nothing prevents us from updating the benchmark state $X_0$ less frequently to decrease the likelihood of manipulation. That is because manipulating $X_0$ over a long period of time increases the inventory risk for an attacker.}

Crucial for our sequencing rule is the observation that any liquidity pool exchange with two tokens satisfies the following duality property: {\em at any state $X \in \mathbb R^2$, it is either the case that (1) any buy order receives a better execution at $X$  than at $X_0$, or (2) any sell order receives a better execution at $X$  than at $X_0$}. Thus at any point during the execution of the orders in a block, as long as the transactions yet to execute are not all of the same type (i.e., not all buy orders or all sell orders), there is at least one order that would be happier to be the next order to execute compared with executing at the beginning of the block.

To be concrete, in defining our {\em Greedy Sequencing Rule}, let $T_1, T_2, \ldots, T_t$ be the execution ordering up to step $t$ of the current block; these are the transactions already added to the execution ordering. To define which transaction $T_{t+1}$ executes at step $t+1$, we simulate what the state $X_t$ would be after $T_1, T_2, \ldots, T_t$ executes and we add the constraint that $T_{t+1}$ must be a buy order (any buy order) if buy orders receive a better execution at $X_t$  than at $X_0$ and a sell order (any sell order) if sell orders receive a better execution at $X_t$  than at $X_0$. It is possible that at step $t+1$, only buy or sell orders are left to execute. In this case, we let the miner sequence the remaining orders as they wish. Importantly, this rule operates on any set of transactions included in a block and does not need knowledge as to whether a transaction comes from a user or the miner. This makes the rule verifiable based purely on information written on the blockchain.

Interestingly, when only buys or sells are left to execute, we will argue the miner is indifferent as to whether or not to include those in the block. That happens because the miner never profits from executing a buy (respectively sell) order after another buy (respectively sell) order, and would instead prefer to execute their transaction before this kind of order. Therefore, if $T_{t+1}$ is a user buy (or sell) order and all transactions that execute after $T_{t+1}$ are also buy (or sell) orders, the miner executes no order of their own after $T_{t+1}$ because they would prefer instead to execute their orders before $T_{t+1}$. Since the miner does not choose to subsequently include any of their own transactions once only buys or only sells of others remain to execute, then this implies that the miner does not gain from placing these  transactions, for example $T_{t+1}$, in the block. That is, there is no risk-free gain to the miner from including these transactions, i.e., no gain in tokens to the miner.

Let us see why the Greedy Sequencing Rule makes sandwich attacks unprofitable on Uniswap when there is a single user who wishes to purchase $q$ units of token $1$ at market price (i.e., without reporting a limit on how much they would pay) and the initial state in the block is $(X_1, X_2)$. For the setting where the miner can sequence orders as they wish, the miner can obtain a risk-free profit by first front-running the user and purchasing $w < X_1 - q$ units of token $1$. Then they execute the user's order at state $\left(X_1 - w, \frac{X_1 \cdot X_2}{X_1 - w} - X_2\right)$. After the user's order executes, the miner sells the $w$ units of token $1$ they purchased in the first step (at a higher price).

On the other hand, if the miner commits to implementing the Greedy Sequencing Rule, once the miner purchases $w$ units of token $1$, the miner is forced to execute any outstanding sell order (before executing the user's buy order). Thus the sequencing rule forces the miner to immediately sell the $w$ units of token $1$ they just purchased! One can easily check that no matter how many transactions the miner injects into the block, once constrained by the Greedy Sequencing Rule, the miner cannot obtain a risk-free profit when including only a single user transaction in the block. Interestingly, the miner can  obtain risk-free profits if the block contains three or more user transactions (as our impossibility result suggests), but without violating the guarantees of Theorem~\ref{thm:greedy}; i.e., any gains to the miner do not come at the expense of poorer execution price to users. Our work formalizes this intuition, giving results for any two-token liquidity pool exchange and for any number of user transactions included in a block.\footnote{With a single user transaction in the block, our greedy sequencing rule ensures the miner cannot obtain risk-free profits. Our impossibility constructs an attack on any sequencing rule when blocks contain three or more user transactions. We leave as an open question the case where blocks have two user transactions.}

\subsection{Related work}
\noindent{\bf Blockchain consensus.} Decentralized exchanges are one of the most impactful applications of decentralized blockchain technology.~\citet{nakamoto2008bitcoin} introduced the {\em Bitcoin} digital currency as the first use case for decentralized blockchains. The bitcoin blockchain uses a PoW ({\em Proof-of-Work}) longest-chain blockchain to implement a decentralized distributed computer for payments. No single entity owns the bitcoin system because anyone can volunteer to be a miner. The first miner to solve a computationally hard problem receives the privilege to change the state of the distributed computer and receive cryptocurrency in the form of Bitcoin tokens as a reward. Economic incentives play an important role in the security of decentralized blockchains. A line of work started by \citet{eyal2014majority} introduces selfish mining as a way for miners to improve their profit on longest chain PoW blockchains. \citet{sapirshtein2016optimal} and ~\citet{kiayias2016blockchain} provide guarantees for when  honest mining is a Nash equilibrium.

The assumption that miners are myopic---they do not manipulate prices across multiple blocks---is rooted on the assumption that a decentralized blockchain uses an unpredictable miner selection. Although longest chain {\em Proof-of-Work} (PoW) blockchains satisfy this assumption, they have very high energy consumption~\cite{cbeci}. On the other hand, longest chain {\em Proof-of-Stake} (PoS) blockchains~\cite{chen2019algorand, kiayias2017ouroboros, daian2019snow, king2012ppcoin} have a negligible energy cost, but their miner selection are sometimes predictable~\cite{brown2019formal}. Even for non-longest chain PoS blockchains, \citet{ferreira2022optimal} show that an adversary can bias the miner selection making the protocol predictable to a certain degree. \citet{ferreira2021proof} ask if longest chain PoS blockchains can provide similar miner selection guarantees as longest chain PoW. They show that when the blockchain has access to an external source of randomness (such as the NIST randomness beacon) longest chain PoS blockchains can provide similar (but strictly weaker) fairness guarantees (i.e., unpredictable and unbiased miner selection) than their PoW equivalent.

Incentive analysis in blockchain consensus often assumes a constant reward per block~\cite{ferreira2021proof, ferreira2022optimal, kiayias2016blockchain, sapirshtein2016optimal, eyal2014majority}. \citet{carlsten2016instability}, on the other hand, argue that transaction fees, when larger than block rewards, introduce a high variance in the revenue per block and can pose a risk to  blockchain security. \citet{qin2022quantifying} argue that DeFi applications can also disrupt miner incentives. They measure miner extractable value (MEV) from DeFi applications and quantify their risk to the blockchain security.


\vspace{1mm}\noindent{\bf Constant product automated market makers.} Uniswap is the highest trading volume liquidity pool exchange and uses the product potential~\eqref{eq:uniswap}. Their exchange is commonly referred as a constant product automated market maker. There is an underlying risk for providing liquidity to these exchanges, but liquidity providers receive trading fees as compensation. \citet{neuder2021strategic} and~\citet{heimbach2022risks} show that complex liquidity provision strategies can improve the liquidity provider's revenue and~\citet{fan2022differential} studies the tradeoffs between return to liquidity providers and gas fees to traders in the design of Uniswap v3 style schemes for differential price liquidity provision.

\vspace{1mm}\noindent{\bf Miner extractable value.} Our work assumes the miner is profit seeking. Alternatively, front-running on decentralized exchanges have been studied in the context of an honest miner and a self-interested user that attempts to front-run other users~\cite{heimbach2022eliminating, zhou2021high, kulkarni2022towards}. This adversarial model is strictly weaker than ours because a self-interested user has uncertainty over the execution ordering. 

\citet{heimbach2022eliminating} observe that, with trading costs (e.g., trading fees), users can limit their trading volume driving front-running schemes unprofitable. However, their approach is, in effect, limited to a small number of transactions in a block. Otherwise, the adversary can combine multiple transactions by executing them in sequence. For example, if $n>1$ buy orders each have volume $q>0$, then the adversary executes all $n$ orders in sequence which, for all purposes, is equivalent to a single order of volume $n \cdot q$. Thus, there is a sufficiently large $n$ where a front-running scheme remains profitable while unprofitable if executed only on individual transactions. On the other hand, our approach works even if the miner is profit seeking, the number of user transactions per block is unbounded, the trading volume for any particular transaction is arbitrarily large, and there are no trading costs.

\vspace{1mm}\noindent {\bf Mechanism design with imperfect commitment.} Traditional mechanism design assumes that the entity running the mechanism can commit to  the rules of the game. Front-running schemes would not be a concern if the miner could commit to ordering transactions in the same order as they were observed, i.e., without introducing their own transactions after observing user transactions and interspersing them with suitably ordered user transactions. Unfortunately, one cannot enforce such a sequencing rule because it is not verifiable: in the presence of latency, different miners could observe transactions in different orders~\cite{kelkar2020order}. Then the miner has plausible deniability to act on privileged information---and include their own transactions after learning about the user transactions.

Away from the design of mechanisms for decentralized exchanges, this challenge with imperfect commitment is  an important constraint in the design of transaction fee mechanisms for blockchains~\cite{ferreira2021dynamic, lavi2022redesigning, roughgarden2021transaction}. These are the mechanisms that determine which transactions win the right to be executed on a blockchain and enter a block. In auction theory, the inability of the auctioneer to commit to implementing a particular auction rule has been studied through the theory of {\em credible auction design}~\cite{akbarpour2020credible}. This considers ways in which an auctioneer might usefully deviate from an intended rule, but only allowing for deviations that are undetectable by the participants of an auction. For example, an auctioneer can introduce their own bid in a second-price auction to increase the second price, but cannot charge a winner more than their bid price in a first-price auction (and, first-price but not second-price auctions are credible). By running an auction over the internet, it is easy for auctioneers to deviate from the promised auction by, for example, bidding on their own auction with a fake identity. Prior work~\cite{ferreira2020credible, essaidi2022credible, chitra2023credible} has propose computationally and communication-efficient auctions that are truthful and credible under standard cryptographic assumptions and assumptions on bidder valuations.

\subsection{Paper organization}

Our results are not limited to exchange designs such as Uniswap that make 
use of the constant product potential, and  apply to a large class of liquidity pool
decentralized exchanges. We provide the necessary background and introduce a general model for liquidity pool decentralized exchange in Section~\ref{sec:background}. In Section~\ref{sec:model}, we introduce the communication model. In Section~\ref{sec:manipulation}, we show front-running schemes can be profitable for a large class of decentralized exchanges. There we also motivate our impossibility result (Theorem~\ref{thm:impossibility}). In Section~\ref{sec:sequencer}, we define the {\em Greedy Sequencing Rule} and prove our main result, Theorem~\ref{thm:greedy}. We conclude in Section~\ref{sec:conclusion}. Appendix~\ref{app:math} contains the necessary mathematical background. The remaining appendices contain omitted proofs.

\section{Background}\label{sec:background}

The exchange has a {\em state} $X = (X_1, X_2)$ where $X_i \geq 0$ is the current reserves of tokens $i \in \{1, 2\}$. Let $\{e_1, e_2\}$ be the {\em standard basis} of $\mathbb R^2$ which allow us to rewrite $X = X_1 \cdot e_1 + X_2 \cdot e_2$.

A user submits a transaction that either buys or sells token $1$. A buy order $\Buy(q, p)$ purchases $q$ units of token $1$ for at most $p \cdot q$ units of token $2$. A sell order $\Sell(q, p)$ sells $q$ units of token $1$ for at least $p \cdot q$ units of token $2$. We refer to $p$ as the {\em limit price}. An order is a {\em market order} if $p = \infty$ for a buy order and $p = 0$ for a sell order, and for a market order we omit $p$ and write $\Buy(q) := \Buy(q, \infty)$ and $\Sell(q) := \Sell(q, 0)$.

To define the outcome of a transaction, we endow the exchange with a {\em potential function} $\phi : \mathbb R_{\geq 0}^2 \to \mathbb R_{\geq 0}$, which is a real-valued continuous function that takes a state $X$ and maps to the potential $\phi(X) \geq 0$. We assume $\phi$ is {\em strictly increasing} and {\em quasiconcave} as follows:

\begin{definition}[Increasing function]
For a function $f$ we refer to $\dom(f)$ as the domain of $f$. For $x, y \in \mathbb R^n$, we write $x \geq y$ to denote $x_i \geq y_i$ for all $i \in [n]$. A real-valued function $f$ is {\em increasing} if for all $x, y \in \dom(f)$, we have that $x \geq y$, $f(x) \geq f(y)$. Moreover, $f$ is {\em strictly increasing} if for all $x, y \in \dom(f) \subseteq \mathbb R^n$ such that $x \geq y$ and $x_i > y_i$ for some $i \in [n] = \{1, 2, \ldots, n\}$, we have that $f(x) > f(y)$.
\end{definition}

\begin{definition}[Convex Set]
A set $D \subseteq \mathbb R^n$ is {\em convex} if for all $x, y \in D$ and $\alpha \in [0, 1]$, the linear combination $\alpha \cdot x + (1-\alpha) \cdot y \in D$.
\end{definition}

\begin{definition}[Quasiconcave Function]
A real-valued function $f$ is {\em quasiconcave} if $\dom(f)$ is a convex set and for all $x, y \in \dom(f)$, and all $\alpha \in [0, 1]$, we have that $f(\alpha \cdot x + (1-\alpha) \cdot y) \geq \min\{f(x), f(y)\}$.
\end{definition}

The {\em execution price} of an order is the buying price in the case of a buy order, or the selling price in the case of a sell order. We define the execution price algorithmically using the potential function and the current state. For a buy order $\Buy(q)$ executing at state $X$, the function $Y(X, \Buy(q))$ denotes the amount of token $2$ a user would pay for $q$ units of token $1$, which we define as
\begin{align}\label{eq:payment}
    Y(X, \Buy(q)) := \min\{y \geq 0: \phi(X - q \cdot e_1 + y \cdot e_2) \geq \phi(X)\}.
\end{align}

For a sell order $\Sell(q)$ executing at state $X$, the function $Y(X, \Sell(q))$ denotes the amount of token $2$ a user would trade for $q$ units of token $1$, which we define as
\begin{align}\label{eq:withdrawn}
    Y(X, \Sell(q)) := \max\{y \leq X_2 : \phi(X + q \cdot e_1 - y \cdot e_2) \geq \phi(X)\}.
\end{align}

Although $Y(X, \Buy(q))$ (or $ Y(X, \Sell(q))$)  define the execution price at state $X$, we introduce some feasibility constraints to determine if an order will successfully execute or fail. First, an order must not turn the liquidity reserves negative. Second, the potential at the next state must be the same as the previous state. Third, the user must pay at most $q \cdot p \cdot e_2$, in the case of a buy order $\Buy(q, p)$, or receive at least $q \cdot p \cdot e_2$, in the case of a sell order $\Sell(q, p)$. Formally, a buy order $\Buy(q, p)$ can {\em successfully execute at $X$} if and only if,
\begin{align*}
    &Y(X, \Buy(q)) \leq p \cdot q,\\
    &\phi(X_1 - q, X_2 + Y(X, \Buy(q))) = \phi(X),\\
    &X_1 \geq q.
\end{align*}

A sell order $\Sell(q, p)$ can {\em successfully execute at $X$} if and only if
\begin{align*}
&Y(X, \Sell(q)) \geq p \cdot q,\\ &\phi(X_1 + q, X_2 - Y(X, \Sell(q)) = \phi(X),\\
&X_2 \geq Y(X, \Sell(q)).
\end{align*}

If $\Buy(q, p)$ can successfully execute at $X$, the user trades $Y(X, \Buy(q)) \cdot e_2$ for $q \cdot e_1$. Similarly, if $\Sell(q, p)$ can successfully execute at $X$, the user trades $q \cdot e_1$ for $Y(X, \Sell(q)) \cdot e_2$. We summarize the order of operations for the execution of an order on state $X_{t-1}$ in Algorithm~\ref{alg:execution}.

\newpage
\begin{flushleft}
\begin{framed}
\begin{center} \bf Order Execution \end{center}
\ \\
{\bf Input:} Current state $X_{t-1}$; order $A = \Buy(q, p)\, |\, \Sell(q, p)$.

{\bf Output:} Next state $X_t$.

\begin{enumerate}
    \item If $A$ cannot successfully execute at $X_{t-1}$, {\em abort} the execution of $A$. The subsequent state is $X_t = X_{t-1}$.
    
    \item If $A$ can successfully execute at $X_{t-1}$:
    \begin{enumerate}
        \item If $A$ is a buy order, the user deposits $Y(X_{t-1}, A)$ units of token $2$ and withdraws $q$ units of token $1$.
        
        \item If $A$ is a sell order, the user deposits $q$ units of token $1$ and withdraws $Y(X_{t-1}, A)$ units of token $2$. 
        
        \item The subsequent state is 
        $$X_t = \begin{cases} 
        X_{t-1} - q \cdot e_1 + Y(X_{t-1}, A) \cdot e_2 \qquad &\text{if $A$ is a buy order,}\\
        X_{t-1} + q \cdot e_1 - Y(X_{t-1}, A) \cdot e_2 \qquad &\text{if $A$ is a sell order.}
        \end{cases}$$
    \end{enumerate}
\end{enumerate}
\end{framed}
\captionof{algorithm}{The execution of a $\Buy(q, p)$ or $\Sell(q, p)$ order at state $X_{t-1}$.} 
\label{alg:execution}
\end{flushleft}

One benefit of liquidity pool exchanges is their computational efficiency, since Algorithm~\ref{alg:execution} executes in constant time for many choices of $\phi$. For example, computing $Y(X_{t-1}, A)$ for the product potential function of Uniswap requires only a constant number of algebraic operations.

Observe a transaction only successfully executes at state $X$ if the next state has the potential $\phi(X) = c$. Thus the exchange will always be in a state contained in the level set $L_c(\phi)$ which we refer as the collection of {\em reachable states}.

\begin{definition}[Level sets]
Let $c \in \mathbb R$ be a constant. A {\em level set} $L_c(f)$ of real-valued function $f$ is the collection of points $x \in \dom(f)$ where $f(x) = c$. A {\em superlevel set} $S_c(f)$ of $f$ is the collection of points $x \in \dom(f)$ such that $f(x) \geq c$.
\end{definition}

\subsection{Examples of potential functions}\label{sec:example}

This section provides formal definitions for some of the potential functions that are used in practice. However, our results  hold for a larger class of potential functions, of which the ones defined here are illustrative.

Uniswap~\cite{adams2021uniswap} uses a product potential function to implement a liquidity pool exchange with two tokens. Balancer~\cite{martinelli2019non} uses a similar design but supports pools with two or more tokens. 

A concern is that liquidity pool exchanges based on product potentials can have high price volatility when the reserves are small relative to trade volumes. To address this concern, Curve~\cite{egorov2019stableswap} uses a potential function that aims to provide lower price volatility by assuming token prices are stable.
We describe these models next.
\medskip

\noindent\textbf{Product potential.}  The {\em product potential function} $\phi$ maps a state $X$ to the product of the current deposits
\begin{equation}\label{eq:product}\phi(X) = X_1 \cdot X_2.\end{equation}

An exchange with the product potential is also called {\em constant product automated market maker} (CPAMM), referencing the fact the product of the reserves is invariant while liquidity providers neither add nor remove liquidity. 
In a CPAMM, the prices implied by the current state $X$ remain in rough correspondence with the trading prices in a secondary market such as Coinbase. For example, suppose each unit of token $i$ is worth $p_i$ U.S.~dollars in the secondary market. We say $p_i/p_j$ is the {\em relative price} of token $i$ with respect to token $j$. Then we say a CPAMM is {\em in equilibrium} if $X_1/X_2 = p_1/p_2$; otherwise, arbitrageurs would have an incentive to trade in the exchange and take profits in the secondary market. See \citet{angeris2020improved, angeris2022optimal} for a discussion on the role of arbitrage in automated market makers.
\medskip

\noindent\textbf{Stable potential.} 
Potential functions that imply small variation in prices for even large trades are popular in liquidity pool exchanges in the case that the underlying tokens are expected to have a stable price. For example, {\em stablecoins} such as USDC and USDT aim to have a 1-to-1 parity with the U.S.~dollar.\footnote{These coins are successful, with a market cap of over 150 billion U.S.~dollars as of 2022.~\cite{coinmarketcap} USDT and USDC are known as collateralized stable coins, and are issued by entities that promise that users are always able to redeem 1 unit of USDT or USDC for 1 U.S.~dollar. This suggests that these stable coins should always have a 1-to-1 parity with the U.S.~dollar. Another class of stablecoins known as an algorithmic stablecoin is not collateralized. Instead, they rely on incentive mechanisms, and these have so far failed to hold their 1-to-1 parity during periods of market turbulence. Notably, the algorithmic stablecoin UST had a market cap of over \$50 billion U.S.~dollars when it collapsed overnight in 2022 due to a bank run.}
The {\em stable potential} achieves this by combining  the product potential function with the {\em additive potential function},  defined as
\begin{equation}\label{eq:additive}\phi(X) = X_1 + X_2.\end{equation}

With just the additive potential function, 
the price for token $1$ with respect to $2$ would always
equal $1$ and the exchange would be unstable because if
the prices in a secondary market are $p_1 \neq p_2$, arbitrageurs would have an incentive to trade the token with the lowest price for the token with the highest price until all liquidity reserves are depleted. 
The stable potential addresses this by interpolating 
between the additive and the product potentials:
\begin{equation}\label{eq:stable}
\phi(X) = \frac{X_1 \cdot X_2}{\left(\frac{X_1 + X_2}{2}\right)^2}(X_1 + X_2) + \left(1-\frac{X_1 \cdot X_2}{\left(\frac{X_1 + X_2}{2}\right)^2}\right)X_1 \cdot X_2.
\end{equation}

To see that~\eqref{eq:stable} interpolates correctly, it suffices to check that $0 \leq \frac{X_1 \cdot X_2}{\left(\frac{X_1 + X_2}{2}\right)^2} \leq 1$. The lower bound is clear since $X_i \geq 0$ for  $i\in \{1,2\}$. The upper bound follows from the {\em AM-GM inequality}, i.e., $X_1 \cdot X_2 \leq \left(\frac{X_1 + X_2}{2}\right)^2$ (Lemma~\ref{lemma:am-gm}). Note the inequalities are also tight since the lower bound is attained whenever  $X_i = 0$  for some $i\in \{1,2\}$,
and the upper bound is attained whenever $X_1 = X_2$.

In the case that an exchange uses the stable potential and the prices in the secondary market are $p_1 = p_2$, arbitrageurs would have an incentive to trade in the liquidity pool whenever $X_1 \neq X_2$. Hence the only equilibrium  has  $X_1 = X_2$, and the stable potential behaves closer to the additive potential, as desired.

\subsection{Model discussion}\label{sec:model-discussion}

This section argues why our assumptions on potential functions are, in essence, without loss. Firstly, the potential must be strictly increasing because this is equivalent to requiring that users should only be able to withdraw tokens from the exchange if they deposit some payment:
\begin{definition}
A potential $\phi$ has {\em non-zero payment} if for all states $X \in \dom(\phi)$, we have that $Y(X, \Sell(0)) = 0$.
\end{definition}

One can check that all the potential functions from Section~\ref{sec:example} when restricted to states $X > 0$ have non-zero payment. The product potential does not have zero payments if the domain contains some state $X$ where $X_i = 0$ for some $i$. To show an exchange satisfies non-zero payment, Lemma~\ref{lemma:zero-payment} shows that it suffices to check that $\phi$ is strictly increasing with the mild assumption that $\dom(\phi)$ is open and upward closed.

\begin{definition}[Open and upward closed sets]\label{def:open}
Let $B_\varepsilon(x) = \{y \in \mathbb R^n : ||x - y|| \leq \varepsilon\}$ be the a ball around $x$. A set $D \subseteq \mathbb R^n$ is {\em open} if for all $x \in D$, there is $\varepsilon > 0$ such that for all $y \in B_{\varepsilon}$, we have that $y \in D$. A set $D \subseteq \mathbb R^n$ is {\em upward closed} if for all $x \in D$, if $y \geq x$, then $y \in D$.
\end{definition}

\begin{lemma}\label{lemma:zero-payment}
Let $\dom(\phi)$ be open and upward closed. A potential $\phi$ has non-zero payment if and only if $\phi$ is strictly increasing.
\end{lemma}

\begin{proof}
We prove the lemma in two parts. 
\begin{claim}
If $\phi$ is strictly increasing, then $\phi$ has non-zero payment.
\end{claim}
\begin{proof}
Fix any $X \in \dom(\phi)$. Fix $\varepsilon > 0$ and $X' = X - y \cdot e_2$ for any $y \in (0, \varepsilon]$. Because $\dom(\phi)$ is open, there is a $\varepsilon > 0$ such that $X' \in \dom(\phi)$. Because $\phi$ is strictly increasing and $X$ is strictly bigger than $X'$, we conclude $\phi(X') < \phi(X)$. This proves $Y(X, \Sell(0)) \leq y \leq \varepsilon$. Taking the limit as $\varepsilon \to 0$ proves that $Y(X, \Sell(0)) = 0$ as desired.
\end{proof}

\begin{claim}
If $\phi$ has non-zero payment, then $\phi$ is strictly increasing.
\end{claim}
\begin{proof}
Fix any $X \neq X' \in \dom(\phi)$ and w.l.o.g. assume $X \geq X'$. Let $p = X - X' \geq 0$. Define $Z_j = X - \sum_{i = 1}^j p_i \cdot e_i$ for all $j$ and observe $Z_0 = X$ and $Z_2 = Z'$. Note $Z_j \in \dom(\phi)$. To see, observe $Z_j \geq Z_2 = Z'$ and $Z' \in \dom(\phi)$ which combined with the assumption $\dom(\phi)$ is upward closed implies $Z_j \in \dom(\phi)$. Now observe that $\phi(Z_j) < \phi(Z_{j-1})$ for all $j$ where $p_j > 0$; otherwise, the event $\phi(Z_j) \geq \phi(Z_{j-1})$ implies $Y(Z_{j-1}, \Sell(0)) \geq p_j > 0$, a contradiction to the assumption $\phi$ has non-zero payment. Note there is at least one $p_j > 0$ because $X \neq X'$. This proves $\phi(X) = \phi(Z_0) > \phi(Z_2) = \phi(X')$ as desired. Thus $\phi$ is strictly increasing.
\end{proof}
Combining both claims proves Lemma~\ref{lemma:zero-payment}.
\end{proof}

Our second assumption is that potential functions are quasiconcave. In Lemma~\ref{lemma:pricing}, we show that assuming $\phi$ is quasiconcave and strictly increasing ensures that buying token $1$ only increases the price of token $1$ relative to token $2$ and selling token $1$ only decreases the price of token $1$ relative to token $2$.

\begin{lemma}[Pricing Lemma]\label{lemma:pricing}
Consider states $X$ and $X'$ where $\phi(X) = \phi(X')$ and $X_1' < X_1$ and assume the potential function $\phi$ is quasiconcave and strictly increasing. Then the following hold:
\begin{itemize}
\item If $\Buy(q)$ can successfully execute at both $X$ and $X'$, then $Y(X, \Buy(q)) \leq Y(X', \Buy(q))$.

\item If $\Sell(q)$ can successfully execute at both $X$ and $X'$, then $Y(X, \Sell(q)) \leq Y(X', \Sell(q))$.
\end{itemize}
\end{lemma}

We provide the proof of the Pricing Lemma in Appendix~\ref{app:pricing} which follows from first principles.

\section{Communication Model}\label{sec:model}

The risk of market manipulation in liquidity pool exchanges arrives from how users communicate their transactions with the exchange. Suppose  users $1, 2, \ldots, |A|$ want to execute transactions $A_1, A_2, \ldots, A_{|A|}$ at state $X_0$. Note a single entity could control multiple users, but that is not relevant for our analysis. Each user privately sends their transaction to the miner. The miner aggregates observed transactions into a block $B$, which we model as a set of potentially unbounded size.

The order of transactions in the block defines the {\em execution ordering}---the order by which transactions execute in the decentralized exchange. In our model,  miners pick the block (i.e., the transactions to include), but use a {\em sequencing rule} $S$ to determine the execution ordering.

\begin{definition}[Sequencing Rule]
A {\em sequencing rule} $S$ is a function from a state $X$ (of the liquidity reserves before any transaction executes) and a set of transactions $B$ to a non-empty set system $S(X, B)$ containing permutations of $B$.
\end{definition}

First, we would like a sequencing rule to be efficiently computable in order to minimize the computational burden on miners.
\begin{definition}[Efficient Sequencer]
A sequencing rule $S$ is (computationally) {\em efficient}, if for all initial state $X=(X_1,X_2)$ and block $B$, there is an algorithm that takes $(X_0, B)$ and outputs some $T \in S(X_0, B)$ in time $O(\log(X_1 + X_2) |B|)$.
\end{definition}

Any block sequencing algorithm requires at least $\log(X_1 + X_2) |B|$ computation to read the content of $B$. Thus our definition requires that a sequencing rule imposes at most a constant multiplicative computational overhead when compared with the status quo, i.e., where the miner computes their favorite ordering of $B$.

We are ready to define the {\em trading game} $(X_0, \{A_i\}, S)$ between users and a miner. The game takes as input a transaction, $A_i$, from each user $i$, the initial state $X_0$, and a sequencing rule $S$. The outcome of the game is an execution ordering on a set of transactions and associated sequence of states, where the transactions that are ordered can include a subset of user transactions and additional transactions that may be introduced by the miner. In the case of an honest miner, the game proceeds as in Algorithm~\ref{alg:trading}:

\newpage
\begin{flushleft}
\begin{framed}
\begin{center} \bf Ideal Trading Game \end{center}
\ \\
{\bf Input:} Initial state $X_0$; order $A_i$ from each user $i$; sequencing rule $S$.

{\bf Output:} Execution ordering $T = (T_1, \ldots, T_{|T|})$ and states $X_1, X_2, \ldots, X_{|T|}$ where $T_t$ executes at $X_{t-1}$. 

Proceed as follows:
\begin{enumerate}
    \item The miner initializes the block $B = \emptyset$.
    
    \item For all $i$, user $i$ privately sends order $A_i$ to the miner.
    
    \item For all $i$, the miner adds $A_i$ to $B$.
    
    \item The miner picks some $(T_1, \ldots, T_{|B|}) \in S(X_0, B)$ as the execution ordering.
    
    \item For $t = 1, \ldots |B|$, 
    \begin{enumerate}
        \item The exchange executes $T_t$ at state $X_{t-1}$.
        \item Let $X_t$ be the state after $T_t$ executes on $X_{t-1}$.
    \end{enumerate}
\end{enumerate}
\end{framed}
\captionof{algorithm}{Trading game with an honest miner.}
\label{alg:trading}
\end{flushleft}

Algorithm~\ref{alg:trading} assumes the miner commits to implement all the steps faithfully.
However, we assume miners can perform the following deviations in the real trading game, these comprising the {\em strategy} of the miner, with knowledge of the input:
\begin{enumerate}
    \item At step (3), by censoring transaction $A_i$ (not adding $A_i$ to block $B$).
    
    \item Between step (3) and (4), by including their own transactions into $B$.
    
    \item At step (4), by picking an execution ordering $T \not \in S(X_0, B)$.
\end{enumerate}

We assume a self-interested miner who will follow any set of deviations that are profitable and undetectable by an observer. For this, we assume that the 
observer can only see the blockchain state (which includes $X_0$ and the outcome of the trading game). By assuming an observer can see the outcome but not the set of outstanding orders, we avoid introducing (strong) assumptions over the communication channel. For example, a naive approach to mitigate censorship is to suggest that users broadcast their transaction. Thus, an observer who sees order $A_i$ must conclude that the miner also saw $A_i$. If the miner did not include $A_i$ into $B$, then the observer can conclude the miner censored $A_i$. Unfortunately, it is not possible to confirm that a miner receives a transaction in the presence of latency. The miner can simply refuse to acknowledge the receipt of a message, with latency providing plausible deniability. Moreover, a malicious user could send $A_i$ to the observer and not to the miner, harming the miner's reputation. To avoid these concerns, we assume our observer can only rely on information stored in the blockchain to detect a miner deviation.
\begin{definition}[Safe deviation]
The outcome of the trading game $(X_0, \{A_i\}, S)$ with an honest miner is $(X_0, T)$ where $T \in S(X_0, \{A_i\})$. A {\em deviation} for the miner is a strategy that results in an outcome $(X_0, T)$ where $T$ is not necessarily contained in $S(X_0, \{A_i\})$. A deviation resulting in an outcome $(X_0, T)$ is {\em safe} if there is a trading game $(X_0, \{A_i'\}, S)$ where $T \in S(X_0, \{A_i'\})$.
\end{definition}

In other words, a deviation for the miner is safe if it can be explained by a different set of user transactions $\{A_i'\}$. For a given sequencing rule, we define the {\em language},
\begin{equation}\label{eq:language}
\lang(X_0, S) = \{(X_0, T) : \{A_i\}, T \in S(X_0, \{A_i\})\},
\end{equation}
which is the set of all possible outcomes for the sequencer $S$ with initial state $X_0$.
An important constraint for a sequencing rule is that an observer must be able to efficiently check whether $(X_0, T) \in \lang(X_0, S)$, for any outcome $(X_0,T)$. This can be formalized as follows:
\begin{definition}[Verifiable sequencing rule]
A sequencing rule with language $\lang$ is {\em verifiable} if there is a polynomial time algorithm $P$ that takes any outcome $(X_0, T)$ and outputs $\true$ if $(X_0, T) \in \lang(X_0, S)$ or $\false$ if $(X_0, T) \not\in \lang(X_0, S)$.
\end{definition}

\subsection{User and miner utility}
Consider an outcome $(X_0, T)$. Let $T_{\sigma(i)}$ denote the transaction owned by user $i$ in the execution order $T$, which is undefined if the miner censors user $i$. Then, we define {\em user $i$'s utility} for the outcome as
$$u_i(X_0, T) = \begin{cases}(0, 0) & \qquad \text{if $\sigma(i)$ is undefined, and}\\
X_{\sigma(i)} - X_{\sigma(i)-1} & \qquad \text{otherwise.}
\end{cases}$$

This is the difference in tokens owned by the user that correspond to the executed transaction, or no change in the case that the user's transaction is not executed. For an entity that controls users $i$ and $j$, their utility would be $u_i(X_0, T) + u_j(X_0, T)$.

Let $I \subseteq \{1, \ldots, |T|\}$ be the time steps where one of the miner's transaction executes (perhaps empty). Then the {\em miner's utility} is
$$u_0(X_0, T) = \sum_{t \in I} (X_t - X_{t-1}).$$

For outcomes $(X_0, T)$ and $(X_0, T')$, we say {\em $(X_0, T)$ dominates $(X_0, T')$ for agent $i$} if $u_i(X_0, T) \geq u_i(X_0, T')$. The outcome is a {\em risk-free execution} for agent $i$ if $u_i(X_0, T) \geq 0$. A risk-free execution is {\em profitable} for agent $i$  if the agent has  a strictly positive quantity of some token and a non-negative quantity of the other token. A common feature of a front-running scheme is that the miner 
has an execution ordering and set of transactions to insert that provides
it with a profitable, risk-free execution  (with $0 \neq u_0(X_0, T) \geq 0$). 
\begin{example}
Let $u_i(X_0, T)$ be the utility of agent $i$. If $u_i(X_0, T) = (-1, 1)$, then $(X_0, T)$ is not a risk-free execution for agent $i$, because agent $i$ pays $1$ unit of token $1$. If $u_i(X_0, T) = (0, 0)$, then $(X_0, T)$ is a risk-free execution for agent $i$, but not a profitable one. If $u_i(X_0, T) = (1, 0)$, then $(X_0, T)$ is a profitable, risk-free execution for agent $i$.
\end{example}

\section{Market Manipulation}\label{sec:manipulation}

One desirable property for an exchange with the product and stable potential functions is that a user cannot completely deplete token $1$ reserves because they would need to deposit an unbounded quantity of token $2$. We can formalize this property as follows.
\begin{definition}[Liquidity-preserving]
An exchange is {\em liquidity-preserving} if, for all states $X > 0$, we have
$$\lim_{q \to X_1} Y(X, \Buy(q)) = \infty.$$
\end{definition}

Unfortunately, users interacting with any liquidity-preserving decentralized exchange are vulnerable to front-running if the miner can arbitrarily choose an execution ordering.
We first illustrate this in the following example for the product potential. Although the example focuses on buy orders, sell orders can also be victims of front-running.
\begin{example}[Front-running, product potential]
Consider the product potential, $\phi(X) = X_1 \cdot X_2 = c$, and suppose a  user submits a market order $T_u = \Buy(q)$. Recall the set of reachable states is the level set $L_c(\phi)$. Also observe that $X_1$ uniquely determines $X_2$ given $c$ (because $\phi$ is strictly increasing, Lemma~\ref{lemma:generator}). In the front-running scheme from Figure~\ref{fig:front-running}, the miner sequences their own buy order, $T_1$, before $T_u$ and their own sell order, $T_3$, after $T_u$. This is a sandwich attack. The effect of ``running ahead" of the user's transaction is that the price increases after the miner's buy order. As a result, this is a profitable risk-free manipulation. For the user, the effect is a worse execution price.
\label{ex:sand}
\end{example}
\begin{figure}
    \centering
    \includegraphics[width=\linewidth]{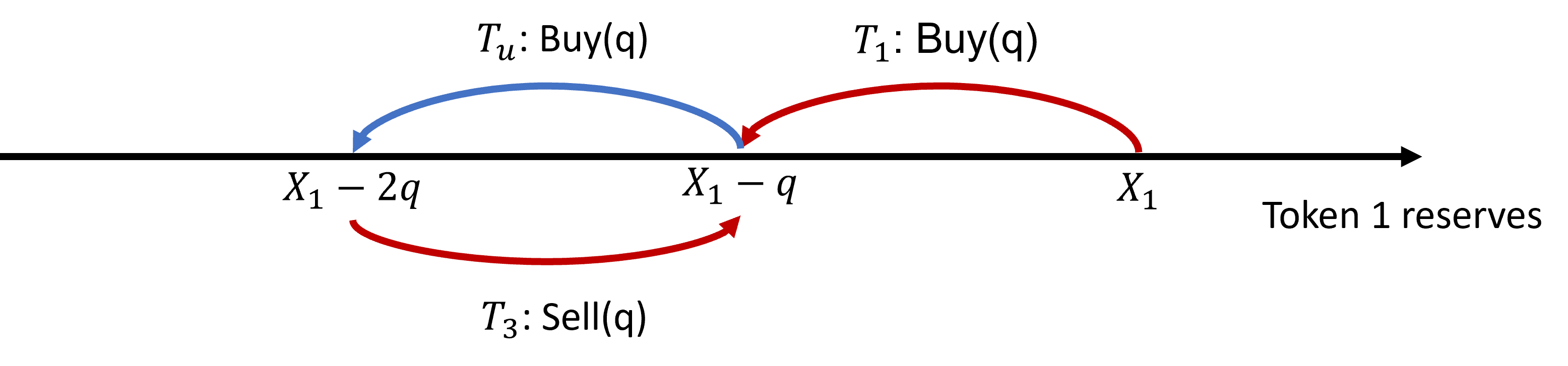}
    \caption{A front-running scheme (a sandwich attack) with execution ordering $(T_1, T_u, T_3)$, where $T_u$ is the user's transaction. The miner's orders,  $T_1$ and $T_3$, 
    are in red, and the user's order, $T_u$, is in blue.}
    \label{fig:front-running}
\end{figure}

Theorem~\ref{thm:pump-dump} generalizes Example~\ref{ex:sand} to the case of a user's transaction that includes a limit price, and for a general exchange rule (not just the product potential), and shows that the miner can obtain a profit that increases as the user increases the limit price. 
\begin{theorem}\label{thm:pump-dump}
Consider a  buy order $\Buy(q, p)$, for quantity $q$ at limit price $p$, and assume this is feasible at state $X$. Assume the exchange is liquidity-preserving. Then there is an execution ordering where the user swaps $q \cdot p$ units of token $2$ for $q$ units of token $1$, and the miner receives $q \cdot p - Y(X, \Buy(q))$ units of token $2$ for free.
\end{theorem}
\begin{proof}
Because $\Buy(q, p)$ is feasible at state $X$, it satisfies the constraints $q \leq X_1$ and $Y(X, \Buy(q, p)) \leq q \cdot p$. Instead of executing only $\Buy(q, p)$ at $X$, the miner injects their own buy and sell by picking the execution ordering
$$(\Buy(X_1-q-w), \Buy(q, p), \Sell(X_1-q-w)).$$
The miner picks the smallest $w \geq 0$ that still allows the user transaction to successfully execute. Observe such a constant exists because setting $w = X_1 - q$ ensures the user transaction successfully executes.

Next, we claim the user pays exactly $q \cdot p$ for this value of $w$. To see this, let $x \geq q$, and observe that the liquidity-preserving assumption implies
$$\lim_{x \to X_1} Y(X, \Buy(x-q)) + Y(X(x), \Buy(q)) = \lim_{x \to X_1} Y(X, \Buy(x)) = \infty,$$
where $X(x)$ is the state after $\Buy(x-q)$ executes on $X$. Thus when $p = w = 0$, the payment of the user can be unbounded since $\lim_{x \to X_1} Y(X(x), \Buy(q)) = \infty$. Moreover, when $w = X_1 - q$, the user pays at most $Y(X, \Buy(q, p)) \leq q \cdot p$. The fact $\phi$ is continuous implies the function representing the user payment as a function of $w$ is continuous. From the intermediate value theorem (Lemma~\ref{thm:intermediate-value}), we conclude there is a value $w \geq 0$ such that the user pays $q \cdot p$ units of token $2$.

The state after all transactions execute is $(X_1 - q, X_2 + Y(X, \Buy(q)))$. Because the user deposits $q \cdot p \cdot e_2$, the miner receives the difference $q \cdot p - Y(X, \Buy(q))$.
\end{proof}

In the case of the additive potential function~\eqref{eq:additive}, which is not liquidity-preserving, the deviation used by the miner in the proof of Theorem~\ref{thm:pump-dump} is not profitable and, in particular, the miner's utility is $(0, 0)$. Whenever $\Buy(q)$ successfully executes, the miner trades $q$ units of token $2$ for $q$ units of token $1$. Equivalently, whenever $\Sell(q)$ successfully executes, the miner trades $q$ units of token $1$ for $q$ units of token $2$. As a result, the miner's net utility from executing orders $\Buy(q)$ and $\Sell(q)$ is always zero regardless of their position in the execution ordering. Although robust to market manipulation, non-liquidity preserving exchanges  fail to execute transactions once  they have exhausted their liquidity reserves, making then unsuitable for settings where secondary-market prices fluctuate.

\subsection{Impossibility result}\label{sec:impossibility}

We already established (Theorem~\ref{thm:pump-dump})  that miners can obtain large profits when they pick the block content along with the execution ordering. Next, we ask what is possible if miners can only choose the block content, but not the execution ordering. One might ask if there is a sequencing rule $S$ where, for any block $\{B_i\}$ (containing miner and user transactions), and initial state $X_0$, there is
no execution ordering $T \in S(X_0, B)$ that gives risk-free profits for the miner.
Unfortunately, Theorem~\ref{thm:impossibility} shows that this is impossible even when $\{B_i\}$ contains only three user transactions. 
For simplicity, we state Theorem~\ref{thm:impossibility} under the assumption that the exchange uses the product potential (which suffices for an impossibility result); however, one can generalize the statement for any liquidity pool exchange that suffers price impact, i.e., the token 1 (resp. 2) price {\em strictly increases} as token 1 (resp. 2) reserves decreases.
\begin{theorem}\label{thm:impossibility}
Consider a liquidity pool exchange with the product potential. Then for any sequencing rule $S$, there is an initial state $X_0$ and a block $B$, containing miner transactions and a set of three user transactions, such that the miner receives a strictly positive quantity of token $2$ and pays nothing if the execution ordering is contained in $S(X_0, B)$.
\end{theorem}

The idea is to consider a block that contains $n \geq 3$ identical buy orders, $\Buy(2)$, and the same number of identical sell orders, $\Sell(1)$, where one of the buy orders  and two of the sell orders are the miner's. We assume  $X_1 > 4n$ in the initial state, so that no transaction ever fails to execute. We fix the execution ordering. Then because the execution ordering from a sequencing rule does not depend on whom owns which transaction, we choose which buys and sells the miner owns after observing the ordering.

It is clear that after the miner executes their orders, their utility in token $1$ is zero. It suffices to argue that for any permutation of the $2n$ transactions, there is always one buy order and two sell orders where the buy order buys at an average price $p$ and the two sell orders sells at an average price that is strictly larger than $p$. Letting the miner own these particular, three transactions, implies profit to the miner because there is a gain in token $2$ and no change in position for token $1$. See Figure~\ref{fig:impossibility} for one of the permutations, where if the miner owns orders $\{T_1, T_4, T_5\}$, they receive a positive quantity of token $2$ for free.
\begin{figure}
    \centering
    \includegraphics[width=\linewidth]{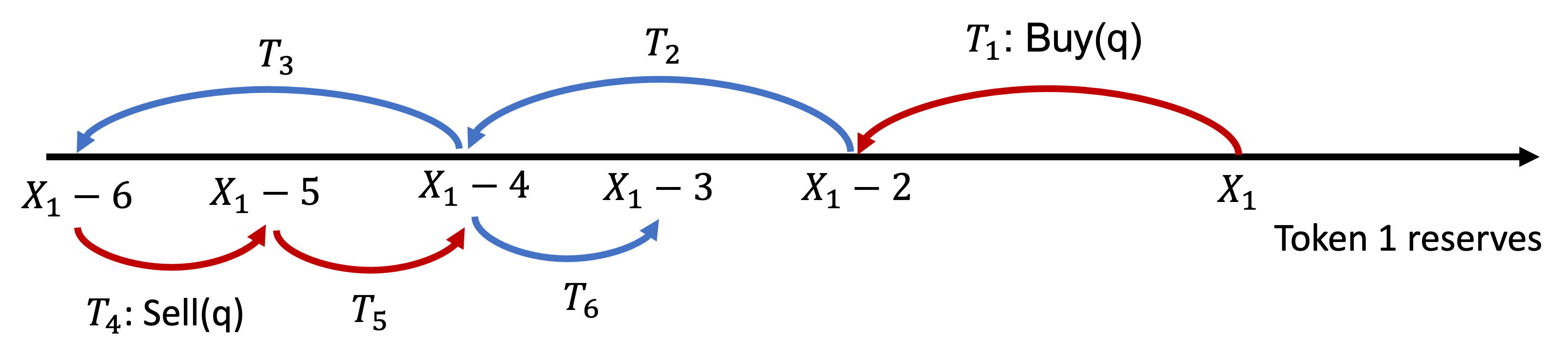}
    \caption{Example of a permutation where the miner obtain risk-free profits. Arrows pointing to the left are buy orders and arrows pointing to the right are sell orders. Miner orders are in red and user orders are in blue.
    \label{fig:impossibility}}
\end{figure}

\begin{proof}[Proof of Theorem~\ref{thm:impossibility}]
Let $X_{0,1} = X_{0, 2} = 4 \cdot n$, where $n \geq 3$. The miner picks a block $B$ that contains $n$ buy orders equal to $\Buy(2)$ and $n$ sell orders equal to $\Sell(1)$. One of the buys and two of the sells will be the miner's.  Let the execution order be $T \in S(X_0, B)$. Next, we will pick which particular buys and sells the miner owns in a way that depends on $T$.

Recall $X_t$ is the state after order $T_t$ executes on $X_{t-1}$. $X_{t, i}$ is the token $i$ reserve at state $X_t$. Let $Z_t = X_{t, 1} - X_{0, 1}$, and observe $Z_0 = 0$ and $Z_{|T|} = -n$ for any execution ordering. Moreover, $Z_t = Z_{t-1} + 1$ whenever $T_t$ is a sell order and $Z_t = Z_{t-1} - 2$ whenever $T_t$ is a buy order.

Let $i \in \arg\max_{t \in [|T|]}\{Z_{t-1} : \text{$T_t$ is a buy order} \}$, then $T_i$ is the buy order that pays the lowest price. Observe $Z_{i-1} = k$ for some $k \geq 0$ because $Z_0 = 0$ and $Z_{|T|} < 0$. Next consider the following cases:

\noindent{\bf Case 1.} {\em There are two sell orders $T_a$ and $T_b$ that execute at a state where $Z_{a-1} \leq (k-2)$ and $Z_{b-1} \leq (k-2)$, respectively.} If the miner owns orders $T_i = \Buy(2)$, $T_a = \Sell(1)$, and $T_b = \Sell(1)$, then the miner ends the trading game with a strictly positive quantity of token $2$ because they sell two units of token $1$ at a lower average price than they were purchased.

\noindent{\bf Case 2.} {\em There are $m \geq n-1$ sell orders that execute at states where $Z_{t-1} \geq k-1$.} Let $T_t$ be one of these sell orders. We first argue that $T_t$ executes at a state where $Z_{t-1} = k-1$. To see this, observe that if $Z_{t-1} > k-1$, then $Z_t = Z_{t-1} + 1 \geq k+1 > 0$. Because $Z_{|T|} = -n < 0$, there is a buy order $T_j$ that executes after $T_t$ where $Z_{j-1} \geq k+1$. We reach a contradiction, since $k = Z_{i-1} \geq Z_{j-1} \geq k+1$. This proves that if $T_t$ is a sell order executing a state where $Z_{t-1} \geq k-1$, then $Z_{t-1} = k-1$. This also implies that, for all time steps $t$, we have  $Z_t \leq k$. Now for any sell order $T_t$ where $Z_{t-1} = k-1$, we will have that $Z_t = k$. Then $T_{t+1}$ must be a buy order, and we will have that $Z_{t+2} = k-2$. Therefore, we can only have $m$ sell orders executing at states where $Z_{t-1} = k-1$ if there are $m$ sell orders executing at states where $Z_j = k-2$. Because $m = n-1 \geq 2$, we reach a contradiction to the assumption that at most one sell order executes at a state where $Z_{t-1} \leq k-2$. This proves 
that Case 2 never happens.

Cases 1 and 2 cover all scenarios. This proves that the miner always receives a strictly positive quantity of token $2$ and pays nothing.
\end{proof}

\section{Greedy Sequencing Rule}\label{sec:sequencer}

The Greedy Sequencing Rule (Algorithm~\ref{alg:greedy})
takes a set of transactions $B$ and an initial state $X_0$ (denoting the state before a transaction in this block executes on the chain), and recursively constructs an execution ordering $(T_1, \ldots, T_{|B|})$ (a permutation of the transactions in $B$). The set of transactions $B$ may include both user and miner transactions.
We refer to $(T_1, \ldots, T_t)$ as the execution ordering up to step $t$ and state $X_t$ denotes the state after $(T_1, \ldots, T_t)$ executes on state $X_0$.
We partition the set of outstanding transactions, that is the 
transactions in $B$ that are not in $\{T_1, \ldots, T_t\}$, 
into two groups: $B^\buy$ that contains the outstanding buy orders and $B^\sell$ that contains the outstanding sell orders. At step 0, $B = B^\buy \cup B^\sell$.

What would be a good choice for $T_{t+1}$? We suppose that when  user $i$ communicated their order $A_i$ to the miner, the user could observe $X_0$ as the current state of the exchange, and is  comfortable with $A_i$ executing at $X_0$. 
However, $A_i$ will execute at state $X_t$ if this forms 
the $(t+1)$-th transaction in the block. If $A_i$ is a buy order, the user prefers $Y(X_t, A_i)$ to be as small as possible. However, if $A_i$ is a sell order, the user prefers $Y(X_t, A_i)$ to be as large as possible. Our main observation is that as long as both $B^\buy$ and $B^\sell$ are not empty, there is at least one transaction that would be at least as happy when executing at $X_t$ as when executing at $X_0$. To be concrete, we show that if the  token 1 reserves at $X_t$ are smaller than those at 
$X_0$, then any sell order would be at least as happy by executing at $X_t$  than $X_0$. Conversely, if the $X_t$ token 1 reserves are higher than at $X_0$, then any buy order would be at least as happy by executing at $X_t$  than $X_0$. We formalize this property in the Duality Theorem (Theorem~\ref{thm:duality}) which we prove in Appendix~\ref{app:duality}.

This immediately suggests a good choice for $T_{t+1}$. The Greedy Sequencing Rule allows the miner to pick $T_{t+1}$ to be any buy order from $B^\buy$ if it is the buy orders that prefer to execute at $X_t$  than $X_0$; otherwise, the rule allows the miner to pick $T_{t+1}$ to be any sell order from $B^\sell$. Duality ensures that one of these conditions is satisfied as long as neither $B^\buy$ or $B^\sell$ are non-empty. Once one of $B^\buy$ or $B^\sell$ is empty, we allow the miner to append the remaining transactions in any arbitrary order (just as in the status quo). Interestingly, we will argue that from the moment that $B^\buy$ or $B^\sell$ are empty, the miner has nothing to profit (or lose) from manipulating the execution ordering going forward.

\newpage
\begin{flushleft}
\begin{framed}
\begin{center} \bf Greedy Sequencing Rule \end{center}
\ \\
{\bf Input:} Initial state $X_0$; set of transactions $B$.
    
{\bf Output:} Execution ordering $T$. 
    
    Proceed as follows:
\begin{enumerate}
    \item Initialize $T$ as an empty list.
    
    \item Let $B^\buy \subseteq B$ be the collection of buy orders in $B$. Let $B^\sell \subseteq B$ be the collection of sell orders in $B$.
    
    \item While $B^\buy$ and $B^\sell$ are both non-empty:
    \begin{enumerate}
        \item Let $t = |T|$.
        \item If $X_{t, 1} \geq X_{0, 1}$:
        \begin{enumerate}
            \item Let $A$ be any order in $B^\buy$.
            \item Append $A$ to $T$ and remove $A$ from $B^\buy$.
        \end{enumerate}
        \item Else:
        \begin{enumerate}
            \item Let $A$ be any order in $B^\sell$.
            \item Append $A$ to $T$ and remove $A$ from $B^\sell$.
        \end{enumerate}
        \item Let $X_{t+1}$ be the state after $A$ executes on $X_t$.
    \end{enumerate}
    
    \item If $B^\buy \cup B^\sell$ is non-empty, append the remaining transactions in $B^\buy \cup B^\sell$ to $T$ in any order.
\end{enumerate}
\end{framed}
\captionof{algorithm}{Greedy Sequencing Rule.}
\label{alg:greedy}
\end{flushleft}

\begin{definition}[Execution quality]
An order $A$ receives a {\em better execution} at state $X$  than $X'$ if either:
\begin{itemize}
    \item order $A$ fails to execute on $X'$, or
    \item order $A$ can successfully execute on both $X$ and $X'$ and if $A$ is a buy order, we have that $Y(X, A) \leq Y(X', A)$, else if $A$ is a sell order, we have that $Y(X, A) \geq Y(X', A)$.
\end{itemize}
\end{definition}

\begin{theorem}[Duality Theorem]\label{thm:duality}
Consider any liquidity pool exchange with potential $\phi$. For any pair of states $X, X' \in L_c(\phi)$, either:
\begin{itemize}
    \item any buy order receives a better execution at $X$  than $X'$, or
    \item any sell order receives a better execution at $X$  than $X'$.
\end{itemize}
\end{theorem}

The premise for a front-running scheme is that the miner manipulates the liquidity reserves to force a sell order to execute at a state $X_t$ where the token 1 reserves are higher than at $X_0$ and a buy order to execute at a state $X_t$ where the token 1 reserves are smaller than at $X_0$. In particular, if the miner wants  a sell order to execute at a state $X_t$ where the token 1 reserves are higher than at $X_0$, then the miner intends to execute a buy order immediately after this transaction. The Greedy Sequencing Rule avoids this pattern because it forces the miner's buy order to execute before the sell order.

To illustrate this, we consider an example with a single user order, $\Buy(q)$ (see Figure~\ref{fig:greedy}). Recall for a front-running attack, the miner wants to pick the execution ordering $(\Buy(q'), \Buy(q), \Sell(q'))$ where $\Buy(q')$ and $\Sell(q')$ are the miner's own transactions. The crucial observation is that the Greedy Sequencing Rule would never output this as a valid execution ordering. In fact, with the Greedy Sequencing Rule, the miner can only output the following four execution orderings on these transactions:
$$T^{(1)} = (\Buy(q), \Sell(q'), \Buy(q')),\ \ \mbox{or}$$
$$T^{(2)} = (\Buy(q'), \Sell(q'), \Buy(q)),\ \ \mbox{or}$$
$$T^{(3)} = (\Sell(q'), \Buy(q'), \Buy(q)),\ \ \mbox{or}$$
$$T^{(4)} = (\Sell(q'), \Buy(q), \Buy(q')).$$

In $T^{(1)}$, $T^{(2)}$, and $T^{(3)}$, the miner's orders simply cancel each other out, resulting in no profit. The miner prefers $T^{(3)}$ over $T^{(4)}$ since their buy order executes at a better price. This shows that, given this sequencing rule, the miner is indifferent as to whether $\Buy(q)$ is in the block or not.
\begin{figure}
    \centering
    \includegraphics[width=\linewidth]{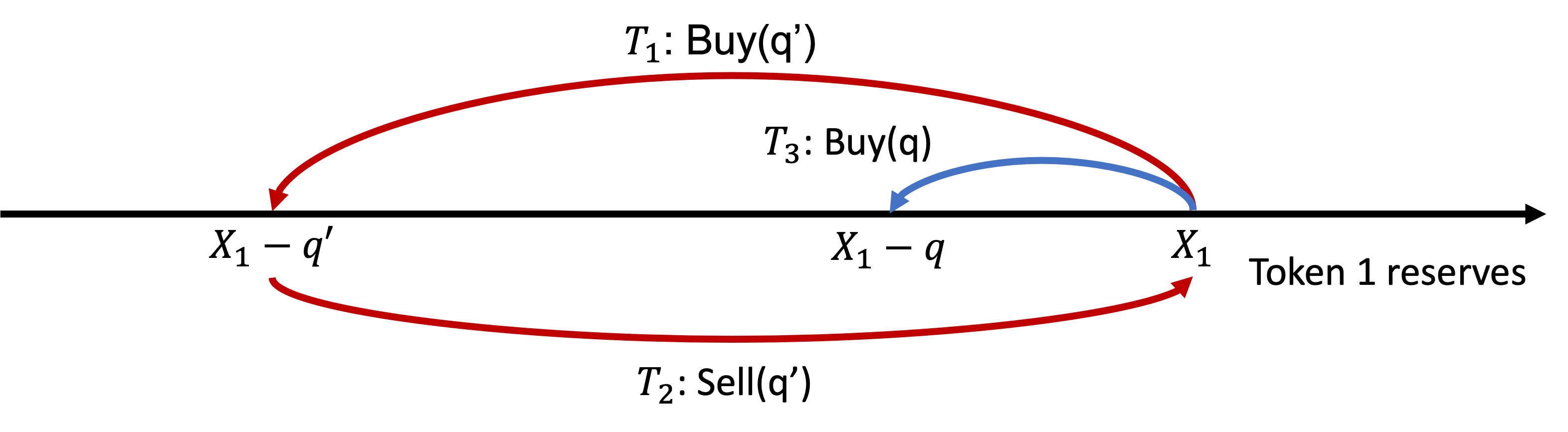}
    \caption{Execution ordering $T^{(1)}$. No  output from the Greedy Sequencing Rule can execute $\Buy(q)$ between $\Buy(q')$ and $\Sell(q')$.
    \label{fig:greedy}}
\end{figure}

\subsection{Execution quality}

In this section, we show that users obtain a predictable execution price that is immune to front-running schemes. Our analysis accounts for an unbounded number of user transactions per block and allows 
the miner to inject as many transactions as they desire.
\begin{definition}[Core/Tail Decomposition]
For an outcome $(X_0, T)$ from any trading game, we partition transactions in $T$ into two sets: the {\em core} denoted $\core(X_0, T)$ and the {\em tail} denoted $\tail(X_0, T)$. A transaction $T_t$ is in the core, if $T_t$ receives a better execution at $X_{t-1}$ than $X_0$. The tail contains the remaining the transactions.
\end{definition}

\begin{lemma}\label{lemma:greedy-tail}
Let $(X_0, T)$ be an outcome from a trading game with the Greedy Sequencing Rule $S$. Then $\tail(X_0, T)$ contains only buy orders or only sell orders.
\end{lemma}
\begin{proof}
The outcome $(X_0, T) \in \lang(X_0, S)$ and there is a corresponding trading game $(X_0, \{T_i\}, S)$ that outputs $(X_0, T)$. During the execution of the trading game, let $T_t$ be an order that is appended to $T$ when neither $B^\buy$ nor $B^\sell$ are empty. We claim $T_t \in \core(X_0, T)$. If $T_t$ is a buy order, it must be that $T_t$ executes at a state $X_{t-1}$ where token 1 reserves are higher than at $X_0$. Consider the following cases:
\begin{itemize}
\item {\em $T_t$ fails to execute at $X_{t-1}$.} From Lemma~\ref{lemma:execution-monotonicity}, if $T_t$ fails to execute at $X_{t-1}$, then it also fails to execute at $X_0$. Thus $T_t$ receives a better execution at $X_{t-1}$  than $X_0$.
\item {\em $T_t$ can successfully execute at $X_{t-1}$.} If $T_t$ fails to execute at $X_0$, then $T_t$ receives a better execution at $X_{t-1}$  than $X_0$. If $T_t$ can successfully execute at $X_0$, the Pricing Lemma (Lemma~\ref{lemma:pricing}) states that $Y(X_{t-1}, T_t) \leq Y(X_0, T_t)$, since token 1 reserves are higher at $X_{t-1}$. Thus $T_t$ receives a better execution at $X_{t-1}$  than at  $X_0$.
\end{itemize}

The above proves if $T_t \in B^\buy$, then $T_t \in \core(X_0, T)$. Next, we consider the case $T_t \in B^\sell$. It must be that $T_t$ executes at a state $X_{t-1}$ where token 1 reserves are lower than at $X_0$. Consider the following cases:
\begin{itemize}
    \item  {\em $T_t$ fails to execute at $X_{t-1}$.} From Lemma~\ref{lemma:execution-monotonicity}, if $T_t$ fails to execute at $X_{t-1}$, then it also fails to execute at $X_0$. Thus $T_t$ receives a better execution at $X_{t-1}$  than $X_0$.
    \item {\em $T_t$ can successfully execute at $X_{t-1}$.} If $T_t$ fails to execute at $X_0$, then $T_t$ receives a better execution at $X_{t-1}$  than $X_0$. If $T_t$ can successfully execute at $X_0$, the Pricing Lemma (Lemma~\ref{lemma:pricing}) states that $Y(X_{t-1}, T_t) \geq Y(X_0, T_t)$, since token 1 reserves are lower at $X_{t-1}$. Thus $T_t$ receives a better execution at $X_{t-1}$  than $X_0$.
\end{itemize}

The above proves that if $T_t \in B^\sell$, then $T_t \in \core(X_0, T)$. Thus all transactions $T_t \in \tail(X_0, T)$ were added to the execution ordering when either $B^\buy$ or $B^\sell$ were empty. This proves that all transactions in $\tail(X_0, T)$ must be of the same type, and only buy orders or only sell orders.
\end{proof}

Let $T_{-t} = (T_1, \ldots, T_{t-1}, T_{t+1}, \ldots, T_{|T|})$ denote the execution ordering $T$ without $T_t$.
\begin{theorem}\label{thm:greedy}
Let $(X_0, T)$ be an outcome from a trading game with the Greedy Sequencing Rule. Then for any user $i \in A$ with a transaction $T_{\sigma(i)} \in T$, we have one of the following
\begin{enumerate}
\item {\bf Indifference.} The execution ordering $T_{-\sigma(i)}$ dominates $T$ for the miner.
    \item {\bf Isolation.}  The execution ordering $T$ dominates $(T_{\sigma(i)})$,
 the  execution ordering that contains only order $T_{\sigma(i)}$, 
    for user $i$.
\end{enumerate}
\end{theorem}
\begin{proof}
Fix a user $i \in A$ with an order $T_{\sigma(i)} \in T$. Consider the following cases:

\begin{itemize}
    \item {\em $T_{\sigma(i)} \in \tail(X_0, T)$.} By Lemma~\ref{lemma:greedy-tail},  all transactions in the tail are of the same type as $T_{\sigma(i)}$. Consider the execution ordering $T_{-\sigma(i)}$, then for any $j > \sigma(i)$, the order $T_j$ receives a better execution under execution ordering $T_{-\sigma(i)}$ than $T$. To see this, observe that if $T_j$ was a buy order, then $T_j$ would execute at a state with higher token 1 reserves (which from Corollary~\ref{cor:execution} only improves the execution of $T_j$). Similar, if $T_j$ was a sell order, then $T_j$ would execute at a state with lower token 1 reserves, which only improves the execution of $T_j$ (Corollary~\ref{cor:execution}). Thus if the miner owns any transaction $T_j$ executing after $T_{\sigma(i)}$, excluding $T_{\sigma(i)}$ only improves the miner's transaction execution quality. This proves that execution ordering $T_{-\sigma(i)}$ dominates $T$ for the miner.
    
    \item {\em $T_{\sigma(i)} \in \core(X_0, T)$.} It follows directly from the definition of the core that $T_{\sigma(i)}$ receives a better execution at state $X_{\sigma(i) - 1}$  than $X_0$. This proves that $T$ dominates $(T_{\sigma(i)})$ for user $i$.
\end{itemize}

The first case argues the miner can only improve their utility by excluding the user's transaction from the block, while the second case argues the user receives an execution as good as the optimal execution they would receive in isolation. This proves the theorem.
\end{proof}

This strong positive result  might come as a surprise given the impossibility result (Theorem~\ref{thm:impossibility}). However, the miner can still profit from injecting their own transactions in the greedy sequencing rule, but provably without causing an execution price worse than what the user expected if sequenced at the start of the block. To see this, consider the same example from Section~\ref{sec:impossibility}, where the block contains three identical buy orders, $\Buy(2)$ and three identical sell orders, $\Sell(2)$,
and the exchange uses the product potential function, $\phi(X) = X_1 \cdot X_2$, 
with initial state $X_1, X_2 \geq 10$. 

Figure~\ref{fig:greedy-example} shows a valid execution ordering from the Greedy Sequencing Rule. By letting the miner own order orders $\{T_1, T_2, T_5\}$, they obtain a positive quantity of token 2 for free. Let us check the execution quality for users. Orders $T_4$ and $T_3$ execute at a state as good as $X$. Order $T_6$ executes at a worse state than $X$; however, the miner is indifferent about $T_6$ (since  no miner transaction is sequenced after $T_6$). Interestingly, if $(T_3, T_4, T_6)$ were the execution ordering, $T_6$ would receive the same execution, although $T_3$ would receive an even better execution because $T_4$ causes a positive externality on $T_3$.
\begin{figure}
    \centering
    \includegraphics[width=\linewidth]{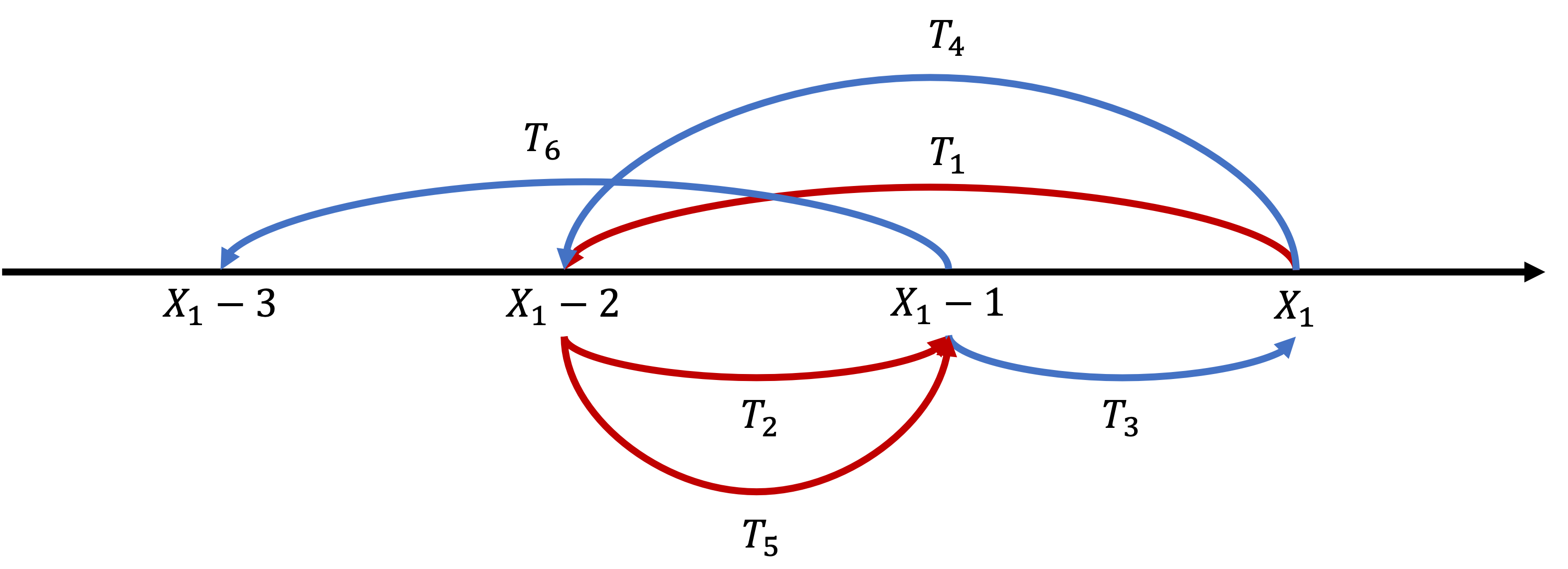}
    \caption{A valid outcome when trading with the Greedy Sequencing Rule for the block that contains three identical buy orders, $\Buy(2)$, and three identical sell orders, $\Sell(2)$. The miner owns the orders in red. The initial state is $(X_1, X_2)$. Left arrows denote buy orders and right arrows denote sell orders.
    \label{fig:greedy-example}}
\end{figure}

\subsection{The Greedy Sequencing Rule is verifiable}

We conclude by providing a proof that the Greedy Sequencing Rule, $S$, is verifiable. Recall that $\lang(X_0, S)$~\eqref{eq:language} is a language representing the feasible outcomes for the sequencing rule $S$ when the initial state is $X_0$. 
We define Algorithm~\ref{alg:verifier} as our verifier. Let us first check that the verifier  outputs {\em True} if $(X_0, T) \in \lang(X_0, S)$ is a valid outcome from the Greedy Sequencing Rule. By definition, there is an input $(X_0, \{A_i\})$ to the Greedy Sequencing Rule that outputs $(X_0, T)$.
Suppose we run the Greedy Sequencing Rule on $(X_0, \{A_i\})$, and suppose for contradiction
the verifier outputs $\false$ at step $t$. The following must be true:
\begin{itemize}
    \item $\{T_t, T_{t+1}, \ldots, T_{|T|}\}$ are not all orders of the same type, and
    
    \item $X_{t-1, 1} \geq X_{0, 1}$ and $T_t$ is a sell order, or $X_{t-1, 1} < X_{0, 1}$ and $T_t$ is a buy order.
\end{itemize}

The first bullet implies $B^\buy$ and $B^\sell$ were not empty at step $t$. Thus whenever $X_{t-1, 1} \geq X_{0, 1}$, the Greedy Sequencing Rule would pick $T_t$ to be a buy order and whenever $X_{t-1, 1} < X_{0, 1}$ the Greedy Sequencing Rule would pick $T_t$ to be a sell order (and since $B^\buy$ and $B^\sell$ are non-empty, this is always possible). This contradicts the second bullet, and proves the verifier outputs $\true$.

Next, consider the case where $(X_0, T) \not \in \lang(X_0, S)$, and this is not a valid outcome from the Greedy Sequencing Rule. Suppose for contradiction, the verifier outputs $\true$. By inspection, we can construct a trading game instance $(X_0, \{A_i\}, S)$ that outputs $(X_0, T)$. Thus $T \in S(X_0, \{A_i\})$ implies $(X_0, T) \in \lang(X_0, S)$, which is a contradiction.

\begin{flushleft}
\begin{framed}
\begin{center} \bf Verifier for the Greedy Sequencing Rule \end{center}
\ \\
{\bf Input:} Outcome $(X_0, T)$.

{\bf Output:} $\true\, |\, \false$.

Proceed as follows:
\begin{enumerate}
    \item For $t = 1, 2, \ldots, |T|$:
    \begin{enumerate}
        \item If $T_t, T_{t+1}, \ldots, T_{|T|}$, are orders of the same type (i.e., all are buy orders or all are sell orders), then output $\true$.
        \item If $X_{t-1, 1} \geq X_{0, 1}$ and $T_t$ is a buy order, then output $\false$.
        \item If $X_{t-1, 1} < X_{0, 1}$ and $T_t$ is a sell order, then output $\false$.
        \item Let $X_t$ be the state after $T_t$ executes on $X_{t-1}$.
    \end{enumerate}
    
    \item Output $\true$.
\end{enumerate}
\end{framed}
\captionof{algorithm}{The Verifier for the Greedy Sequencing Rule.}
\label{alg:verifier}
\end{flushleft}

\section{Conclusion}\label{sec:conclusion}

We have proposed a decentralized exchange  framework where we model miners  as intermediaries between a liquidity pool exchange and users, acting to  choose which transactions to include in a block. The miners are free to pick the block content, but the execution ordering must be a feasible output from the Greedy Sequencing Rule that we introduce and that they commit to implement. Our design does not require users to trust miners because we require that the execution ordering is verifiable, i.e., a polynomial time algorithm can certify if the execution ordering is a valid output from the Greedy Sequencing Rule. Our proposal is backward compatible with the current  implementations of liquidity pool exchanges: miners incur no computational overhead for implementing the sequencing rule and the sequencing rule verifier runs in polynomial time. From a practical side, our proposal can be operationalized through a relay service that commits to following our sequencing rule and operates a private transaction pool.

A common desideratum in DeFi protocols is to limit the miner utility from protocol manipulations (i.e., removing so-called miner extractable value). In this regard, we prove that no sequencing rule can prevent miners from obtaining risk-free profits in liquidity pool exchanges that have price impact: there are always instances where the miner receives strictly positive utility. Given this impossibility result, we focused on sequencing rules that provide provable guarantees for users in a system populated by self-interested miners. Our main finding is the Greedy Sequencing Rule,  which ensures that for any user transaction $A$ that is executed, either (1) transaction $A$ receives an execution price as good as if $A$ was the only transaction in the block, or (2) transaction $A$ receives an execution price worse than this standalone price, but the miner does not gain when including $A$ in the block. Thus a bad execution price is due to competition for the same token between users, and not the result of manipulation from the miner.

Our work leaves open the following questions:
\begin{itemize}
\item Our Greedy Sequencing Rule explicitly assumes the liquidity pool exchange only pools $n = 2$ tokens. Although two token pools have the highest trading volume, other liquidity pool exchanges allow for liquidity pools with three or more tokens. Namely, Balancer~\cite{martinelli2019non} uses the product potential $\phi(X) = \prod_{i = 1}^n X_i$ where one might have $n \geq 3$ tokens in the same pool. {\em Which guarantees can verifiable sequencing rules provide for $n \geq 3$ token liquidity pool exchanges?}

\item We propose a verifiable sequencing rule with provable execution price guarantees for users under a simple utility model for miners, i.e., we assume that a miner prefers token basket $X$ over $X'$ if $X \geq X'$ (and incomparable otherwise). {\em If the miner's utility is a real-valued function $u(X) \in \mathbb R$, can one characterize the sequencing rule that minimizes the miner's utility over all trading games and safe deviations?}

\item We show that for any sequencing rule, there is a set of user transactions that allow miners to extract risk-free profits. This result shows that welfare loss for users (i.e., the price of anarchy) is inevitable when miners are strategic. {\em Define an optimal sequencing rule, i.e., a rule that maximizes user welfare under optimal miner deviations. Can we characterize the class of optimal sequencing rules?}

\item Our impossibility result only applies to deterministic sequencing rules. A randomized sequencing rule $S$ takes not only the initial state $X$ and the block content $B$, but also a random string $r$ (potentially unknown by the miner before picking $B$) and outputs a set system $S(X, B, r)$ with valid execution orderings. {\em Are there randomized sequencing rules where miners cannot obtain risk-free profits?}

\item In a setting where a user creates two or more transactions $\{B_i, B_j\}$, they might wish to specify a constraint that $B_i$ executes before $B_j$. {\em Can one replicate our results with sequencing rules that must preserve the user's ordering constraints?}
\end{itemize}

\appendix

\section{Mathematical Background}\label{app:math}

\begin{lemma}[AM-GM Inequality]\label{lemma:am-gm}
Let $x_1, x_2, \ldots, x_n \geq 0$. Then $\frac{1}{n} \sum_{i = 1}^n x_i \geq \sqrt[n]{\prod_{i = 1}^n x_i}$.
\end{lemma}

\begin{theorem}[Intermediate Value Theorem]\label{thm:intermediate-value}
Let $f$ be a real-valued continuous function with domain $\dom(f)$ equals to the interval $[a, b]$. If $\min\{f(a), f(b)\} \leq u \leq \max\{f(a), f(b)\}$, then there is a $c \in [a, b]$ such that $f(c) = u$.
\end{theorem}

\begin{lemma}\label{lemma:quasiconcave}
A real-valued function $f$ is quasiconcave if and only if all its superlevel sets are convex sets.
\end{lemma}
\begin{proof}
First, consider the case where $f$ is quasiconcave. Let $S_c$ be a superlevel set of $f$. Then for any $x, y \in S_c(f)$ and $\alpha \in [0, 1]$, quasiconcavity implies $f(\alpha x + (1-\alpha)y) \geq \min\{f(x), f(y)\} \geq c$ where the last inequality follows from the fact $x$ and $y$ are in the superlevel set $S_c(f)$. The inequality witnesses that the convex combination of $x$ with $y$ also belongs to the superlevel set $S_c(f)$. This proves $S_c(f)$ is a convex set.

Next, consider the case where all superlevel sets $S_c(f)$ of $f$ are convex sets. Then for any $x, y \in \dom(f)$ and $\alpha \in [0, 1]$, we can pick $c = \min\{f(x), f(y)\}$. Then convexity implies the convex combination $\alpha x + (1-\alpha) y \in S_c(f) \subseteq \dom(f)$. Thus $f(\alpha x + (1-\alpha) y) \geq c = \min\{f(x), f(y)\}$. This proves that $f$ is quasiconcave.
\end{proof}

\begin{definition}[Graph and Epigraph]
The {\em graph} of a real-valued function $f$ is defined as
$$\graph(f) = \{(x, f(x)) : x \in \dom(f)\}.$$

A function $f$ {\em generates} a set $S$, if $S = \graph(f)$. The {\em epigraph} of a real-valued function $f$ is defined as
$$\epi(f) = \{(x, y) : x \in \dom(f): y \geq f(x)\}.$$
\end{definition}

\begin{lemma}\label{lemma:epi}
A real-valued function $f$ is convex if and only if $\epi(f)$ is a convex set.
\end{lemma}
\begin{proof}
Consider the case where $f$ is convex. Let $(x_1, y_1), \ldots, (x_m, y_m) \in \epi(f)$, $\sum_{i = 1}^m \alpha_i = 1$. Then
\begin{align*}
    \sum_{i = 1}^m \alpha_i \cdot (x_i, y_i) &\geq \sum_{i = 1}^m \alpha_i \cdot (x_i, f(x_i))\\
    &= \left(\sum_{i = 1}^m \alpha_i x_i, \sum_{i = 1}^m \alpha_i f(x_i)\right)\\
    &\geq \left(\sum_{i = 1}^m \alpha_i x_i, f(\sum_{i = 1}^m \alpha_i x_i)\right) & \qquad \text{\{By convexity of $f$\}}
\end{align*}
The chain of inequalities, proves $\left(\sum_{i = 1}^m \alpha_i x_i, \sum_{i = 1}^m \alpha_i y_i\right) \in \epi(f)$. Thus $\epi(f)$ is convex.

Next, consider the case where $\epi(f)$ is a convex set. Let $x, \ldots, x_m \in \dom(f)$, $\sum_{i = 1}^m \alpha_i = 1$. By definition, $(x_i, f(x_i)) \in \epi(f)$. Then
\begin{align*}
    \left(\sum_{i = 1}^m \alpha_i \cdot x_i, \sum_{i = 1}^m \alpha_i \cdot f(x_i)\right) &= \sum_{i = 1}^m \alpha_i \cdot (x_i, f(x_i)) \in \epi(f)\\
    &\qquad \text{\{By convexity of $\epi(f)$\}}
\end{align*}

This proves $\sum_{i = 1}^m \alpha_i \cdot x_i \in \dom(f)$ and $\sum_{i = 1}^m \alpha_i \cdot f(x_i) \geq f(\sum_{i = 1}^m \alpha_i \cdot x_i)$. Thus $f$ is a convex function.
\end{proof}

\begin{restatement}[Lemma~\ref{lemma:convex-slope}]
If $f : \mathbb R \to \mathbb R$ is a convex function, then the slope function $R$ (Definition~\ref{def:slope}) of $f$ is increasing on all dimensions. That is, for all $x, y, z \in \dom(f)$,

\begin{itemize}
    \item if $x \leq z$, then $R(x, y) \leq R(z, y)$, and 
    \item if $y \leq z$, then $R(x, y) \leq R(x, z)$.
\end{itemize}
\end{restatement}
\begin{proof}
 $R$ is a symmetric function because for all $x, y \in \dom(f)$, and we obtain that $R(x, y) = R(y, x)$. Thus it suffices to show that for fixed $x \in \dom(f)$, the function $h(y) = R(x, y)$ defined over $\dom(f)$ is an increasing function. Pick any point $z \in \dom(f)$ such that $z > y$. Then it suffices to show that $h(y) \leq h(z)$. We consider separately the case where (1) $x \leq y \leq z$, (2) $y \leq x \leq z$ and (3) $y \leq z \leq x$. For each case, we observe that the midpoint can be defined as a convex combination of the extreme points. That is, the first case implies there is an $\alpha \in [0, 1]$ such that $y = \alpha z + (1-\alpha)y$. Then convexity of $f$ implies
$$f(y) \leq \alpha f(z) + (1-\alpha) f(x)$$.

Rearranging the inequality and observing that $\alpha = \frac{y-x}{z-x}$, one obtains that
\begin{align*}
    h(y) = \frac{f(y) - f(x)}{y - x} \leq \frac{f(z) - f(x)}{z - x} = h(z).
\end{align*}

The analysis of the second and third cases is similar. This proves the slope function is an increasing function as desired.
\end{proof}

\section{Proof of Pricing Lemma}\label{app:pricing}

\begin{restatement}[Lemma~\ref{lemma:pricing}]
Consider states $X$ and $X'$ where $\phi(X) = \phi(X')$ and $X_1' < X_1$ and assume the potential function $\phi$ is quasiconcave and strictly increasing. Then the following hold:
\begin{itemize}
\item If $\Buy(q)$ can successfully execute at both $X$ and $X'$, then $Y(X, \Buy(q)) \leq Y(X', \Buy(q))$.

\item If $\Sell(q)$ can successfully execute at both $X$ and $X'$, then $Y(X, \Sell(q)) \leq Y(X', \Sell(q))$.
\end{itemize}
\end{restatement}

Consider Figure~\ref{fig:pricing} where the curve represents the set of reachable states. We will show the assumption the potential function is strictly increasing and quasiconcave implies the curve is the graph of a convex function $f$. Then the payment inequalities will follow from observing the slope of $f$ is decreasing (since $f$ is a convex function).

\begin{figure}
    \centering
    \includegraphics[width=\linewidth]{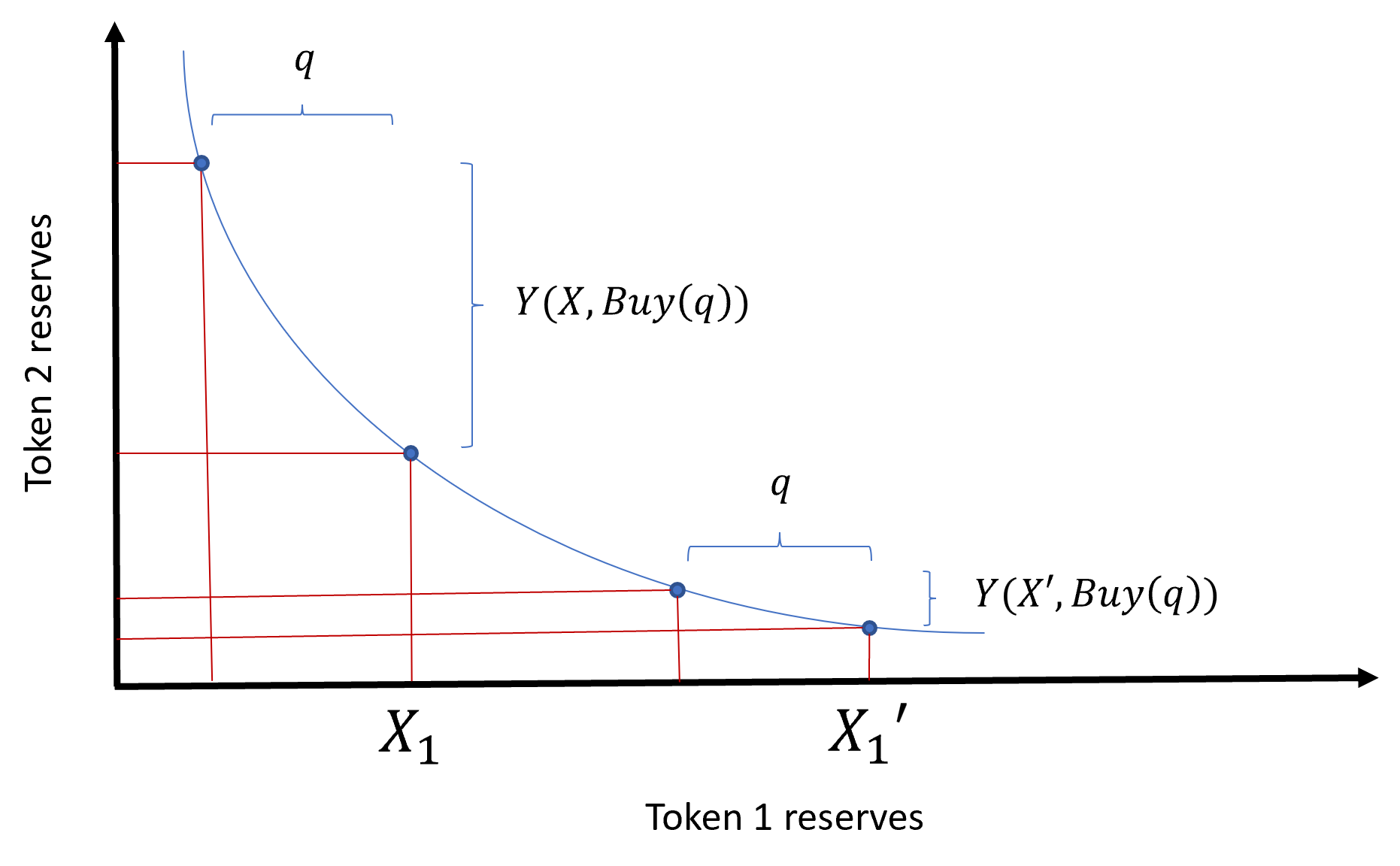}
    \caption{Pricing inequalities for quasiconcave and strictly increasing potential functions.}
    \label{fig:pricing}
\end{figure}

One might wonder if there are two distinct reachable states $(X_1, X_2), (X_1, X_2') \in L_c(\phi)$ or $(X_1, X_2), (X_1', X_2) \in L_c(\phi)$. Interestingly, the fact $\phi$ is strictly increasing precludes this because there is a bijective real-valued function $f$ that generates the set of reachable states $L_c(\phi)$.
\begin{lemma}\label{lemma:generator}
Let $\phi$ be a strictly increasing  potential function. Then there is a bijective real-valued function $f$ that generates level set $L_c(\phi)$ in the following sense: for all $(X_1, X_2) \in L_c(\phi)$, $X_2 = f(X_1)$ and $X_1 = f^{-1}(X_2)$.
\end{lemma}
\begin{proof}
Fix $(X_1, X_2) \in L_c(\phi)$ and $(X_1, X_2') \in \dom(\phi)$. Without loss of generality, let $X_2 > X_2'$. Then the fact $\phi$ is strictly increasing implies $\phi(X_1, X_2) > \phi(X_1, X_2')$. Thus $(X_1, X_2)$ and $(X_1, X_2')$ are not in the same level set. This proves for each $X_1$, there is a unique $X_2$ such that $\phi(X_1, X_2) = c$. Define $f(X_1) = X_2$ for all such $X_1$.

With a similar argument, we can show that for each $X_2$, there is a unique $X_1$ such that $\phi(X_1, X_2) = c$. Define $g(X_2) = X_1$ for all such $X_2$. Finally, we claim $g = f^{-1}$. To see this, observe $g(f(X_1)) = g(X_2) = X_1$. Thus $g$ is the inverse function of $f$, as desired.
\end{proof}

Lemma~\ref{lemma:generator} ensures the existence of a generator $f$ for the set of reachable states $L_c(\phi)$, as long as the potential function $\phi$ is strictly increasing. An interpretation for this result is that the reserves for token $i$ uniquely determine the reserves for token $3$-$i$ (for $i\in \{1,2\}$). That is, if $X_1$ is known, then $X_2 = f(X_1)$. Equivalently, if $X_2$ is known, then $X_1 = f^{-1}(X_2)$. See Figure~\ref{fig:increasing}.

We can now reinterpret Lemma~\ref{lemma:pricing} as stating that the price for token $i$ increases as $X_i$ decreases (since $X_{3-i}$ is completely determined from $X_i$). To prove Lemma~\ref{lemma:pricing}, we will argue that $f$ is a convex function, which  follows from the quasiconcavity of $\phi$. If $\phi$ is quasiconcave, a well-known fact is that the superlevel set $S_c(\phi)$ is a convex set. Interestingly, we can relate $S_c(\phi)$ with $f$ by observing $S_c(\phi) = \epi(f)$, which follows from the fact $\phi$ is strictly increasing. All that is left is to use a well-known property for convex functions:  a real-valued function $f$ is convex if and only if $\epi(f)$ is a convex set (Figure~\ref{fig:convex}).

\begin{figure}
    \begin{minipage}{0.5\linewidth}
    \includegraphics[width=\linewidth]{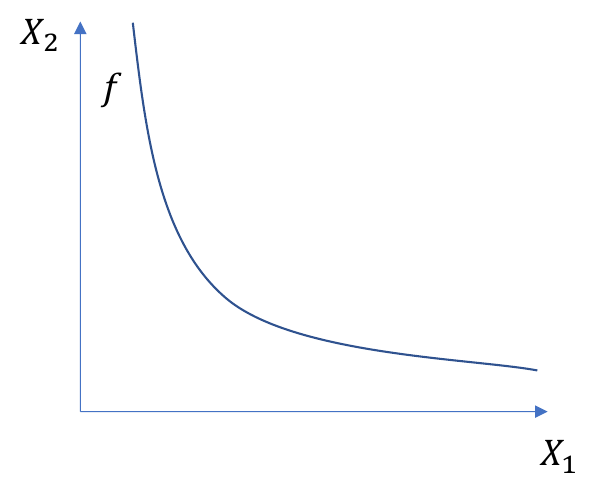}
    \caption{Because $\phi$ is strictly increasing the generator $f$ of level set $L_c(\phi)$ always exists.\\}
    \label{fig:increasing}
    \end{minipage}
    \begin{minipage}{0.05\linewidth}
    \end{minipage}
    \begin{minipage}{0.5\linewidth}
    \includegraphics[width=\linewidth]{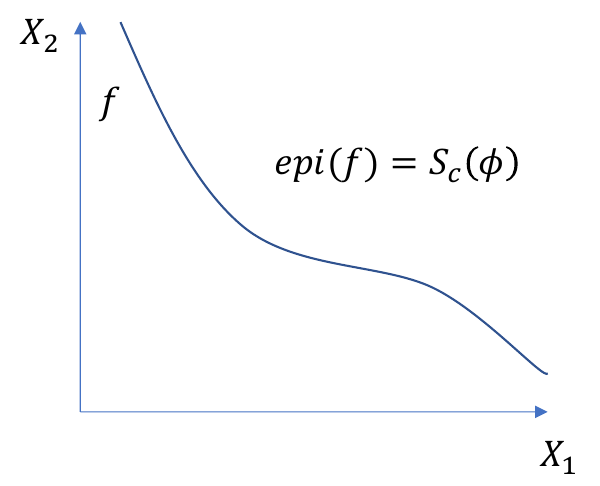}
    \caption{If $\phi$ is strictly increasing and quasiconcave, the generator $f$ of level set $L_c(\phi)$ is a convex function.}
    \label{fig:convex}
    \end{minipage}
\end{figure}

\begin{definition}[Convex function]\label{def:convex-function}
A real-valued function $f$ is convex if $\dom(f)$ is a convex set, and for all $x, y \in \dom(f)$, and all $\alpha \in [0, 1]$, we have that $f(\alpha \cdot x + (1-\alpha) \cdot y) \leq \alpha \cdot f(x) + (1-\alpha) \cdot f(y)$.
\end{definition}

\begin{lemma}\label{lemma:generator-convex}
Let $\phi$ be strictly increasing and quasiconcave. Then there is a bijective convex function $f$ that generates level set $L_c(\phi)$.
\end{lemma}
\begin{proof}
Applying Lemma~\ref{lemma:generator} with strictly increasing function $\phi$ implies the existence of a bijective function $f$ that generates $L_c(\phi)$. We claim $\epi(f) = S_c(\phi)$, for superlevel set $S_c(\phi)$. First, consider the case where $(X_1, X_2) \in S_c(\phi)$. By definition, $\phi(X_1, f(X_1)) = c$ and $\phi(X_1, X_2) \geq c$. Because $\phi$ is strictly increasing, we conclude $X_2 \geq f(X_1)$. This proves $(X_1, X_2) \in \epi(f)$. Next, consider the case where $(X_1, X_2) \in \epi(f)$. By definition, $X_2 \geq f(X_1)$. Because $\phi$ is strictly increasing, $\phi(X_1, X_2) \geq \phi(X_1, f(X_1)) = c$. This proves $(X_1, X_2) \in S_c(\phi)$. Both cases prove $\epi(f) = S_c(\phi)$.

Next, we require two well-known facts from convexity theory (see  Appendix~\ref{app:math}):
\begin{itemize}
    \item 
    A real-valued function is quasiconcave if and only if all its superlevel sets are convex sets.
    
    \item 
    A real-valued function is convex if and only if its epigraph is a convex set.
\end{itemize}

 The first bullet implies $S_c(\phi) = \epi(f)$ is a convex set since $\phi$ is quasiconcave. The second bullet implies $f$ is a convex function, as desired.
\end{proof}

We are ready to prove Lemma~\ref{lemma:pricing}. All that is left is the following well-known fact from convexity theory:
\begin{definition}[Slope function]\label{def:slope}
The {\em slope function} of $f : \mathbb R \to \mathbb R$ maps $x, y \in \dom(f)$ to $R(x, y) = \frac{f(y) - f(x)}{y - x}$, the slope of the line connecting $(x, f(x))$ to $(y, f(y))$.
\end{definition}

\begin{lemma}\label{lemma:convex-slope}
If $f$ is a convex function, then the slope function $R$ of $f$ is increasing on all dimensions. That is, for all $x, y, z \in \dom(f)$,

\begin{itemize}
    \item if $x \leq z$, then $R(x, y) \leq R(z, y)$, and 
    \item if $y \leq z$, then $R(x, y) \leq R(x, z)$.
\end{itemize}
\end{lemma}

We include a proof in Appendix~\ref{app:math}. We conclude with the proof of Lemma~\ref{lemma:pricing}.

\begin{proof}[Proof of Lemma~\ref{lemma:pricing}]
Let $f$ be the generator of level set $L_c(\phi)$ and recall $f$ is a convex function as per Lemma~\ref{lemma:generator-convex}. Using the fact that $X_1' \leq X_1$, we obtain
\begin{align*}
    R(X_1'-q, X_1') &\leq R(X_1-q, X_1') \qquad \text{\{by Lemma~\ref{lemma:convex-slope}\}}\\
    &\leq R(X_1-q, X_1) \qquad \text{\{by Lemma~\ref{lemma:convex-slope}\}.}\\ 
\end{align*}

Expanding on the definition of $R$, we derive
$$f(X_1'-q) - f(X_1') \geq f(X_1-q) - f(X_1).$$
\begin{claim}
Let $A$ be an order that can successfully execute at state $X = (X_1, f(X_1))$. The following are true:
\begin{itemize}
\item if $A = \Buy(q)$, then $Y(X, A) = f(X_1 - q) - f(X_1)$, and
\item if $A = \Sell(q)$, then $Y(X, A) = f(X_1 + q) - f(X_1)$.
\end{itemize}
\end{claim}
\begin{proof}
Let $Z = (Z_1, f(Z_1))$ be the state after $A$ executes at $X$. By definition, $Y(X, A) = f(Z_1) - f(X_1)$. If $A = \Buy(q)$, then $Z_1 = X_1 - q$. If $A = \Sell(q)$, then $Z_1 = X_1 + q$. This proves the claim.
\end{proof}

The claim implies that $Y(X', \Buy(q)) = f(X_1' - q) - f(X_1') \geq f(X_1 - q) - f(X_1) = Y(X, \Buy(q))$, as desired. The proof that $Y(X, \Sell(q)) \leq Y(X', \Sell(q))$ follows a similar format: because $R$ is increasing and $X_1' \leq X_1$, we obtain
$$f(X_1' + q) - f(X_1') = q \cdot R(X_1' + q, X_1') \geq q \cdot R(X_1 - q, X_1) = f(X_1 + q) - f(X_1).$$

Finally, we observe that $Y(X', \Sell(q)) = f(X_1' + q) - f(X_1')$ and $Y(X, \Sell(q)) = f(X_1 + q) - f(X_1)$. This proves Lemma~\ref{lemma:pricing}.
\end{proof}

\section{Proof of Duality Theorem}\label{app:duality}

\begin{restatement}[Theorem~\ref{thm:duality}]
Consider any liquidity pool exchange with potential $\phi$. For any pair of states $X, X' \in L_c(\phi)$, either:
\begin{itemize}
    \item any buy order receives a better execution at $X$  than $X'$, or
    \item any sell order receives a better execution at $X$  than $X'$.
\end{itemize}
\end{restatement}

First, we prove a {\em monotonicity property} for the quality of execution of an order. Consider states $X, X' \in L_c(\phi)$, where the token 1 reserves at $X$ are smaller than at $X'$. We will argue that any buy order receives a better execution at $X'$  than at $X$. Similarly, we will argue that any sell order receives a better execution at $X$  than at $X'$.

\begin{observation}\label{obs:generator-feasibility}
Let $f$ be the generator of $L_c(\phi)$ and $X \in L_c(\phi)$ (Lemma~\ref{lemma:generator}). $\Buy(q)$ can successfully execute at $X$ if and only if $X_1, X_1 - q \in \dom(f)$. $\Sell(q)$ can successfully execute at $X$ if and only if $X_1, X_1 + q \in \dom(f)$.
\end{observation}
\begin{proof}
The fact $X_1 \in \dom(f)$ follows from the fact $f$ generates $L_c(\phi)$ and $X \in L_c(\phi)$. First, we will consider a buy order $\Buy(q)$. If $\Buy(q)$ can successfully execute at $X$, then $(X_1 - q, f(X_1-q)) \in L_c(\phi)$ which implies $X_1 - q \in \dom(f)$. For the converse, if $f(X_1-q) \in \dom(f)$, then $(X_1-q, f(X_1-q)) \in L_c(\phi)$ which implies $\Buy(q)$ can successfully execute at $X$.

Next, we consider a sell order $\Sell(q)$. If $\Sell(q)$ can successfully execute at $X$, then $(X_1 + q, f(X_1+q)) \in L_c(\phi)$ which implies $X_1 + q \in \dom(f)$. For the converse, if $f(X_1+q) \in \dom(f)$, then $(X_1+q, f(X_1+q)) \in L_c(\phi)$ which implies $\Sell(q)$ can successfully execute at $X$.
\end{proof}

\begin{lemma}\label{lemma:execution-monotonicity}
Consider states $X, X' \in L_c(\phi)$ where $X_1 \leq X_1'$. If $\Buy(q, p)$ can successfully execute at $X$, then $\Buy(q, p)$ can successfully execute at $X'$. If $\Sell(q, p)$ can successfully execute at $X'$, then $\Sell(q, p)$ can successfully execute at $X$. 
\end{lemma}
\begin{proof}
Let $f$ be the convex function that generates the set of reachable states $L_c(\phi)$ (Lemma~\ref{lemma:generator-convex}). First, we prove $\Buy(q, p)$ can successfully execute at $X'$. The fact that $X, X' \in L_c(\phi)$ and $\Buy(q, p)$ can successfully execute at $X$ implies $X_1', X_1, X_1 - q \in \dom(f)$ (Observation~\ref{obs:generator-feasibility}). Note $X_1' - q$ is a point between $X_1-q$ and $X_1'$, and thus the convexity of $\dom(f)$ implies $X_1' - q \in \dom(f)$. This proves $\Buy(q)$ can successfully execute at $X'$ (Observation~\ref{obs:generator-feasibility}), and we are left to show $Y(X', \Buy(q)) \leq q \cdot p$. Because the token 1 reserves at $X$ are lower than at $X'$, $\Buy(q)$ pays less by executing at $X'$  than $X$ (Lemma~\ref{lemma:pricing}). Thus the payment by executing at $X'$ is at most $q \cdot p$ since the payment by executing at $X$ is at most $q \cdot p$. This proves $\Buy(q, p)$ can successfully execute at $X'$.

Next, we prove $\Sell(q, p)$ can successfully execute at $X$. The fact that $X, X' \in L_c(\phi)$ and $\Sell(q, p)$ can successfully execute at $X'$ implies $X_1, X_1', X_1' + q \in \dom(f)$ (Observation~\ref{obs:generator-feasibility}). Note that $X_1 + q$ is a point between $X_1$ and $X_1' + q$, and thus the convexity of $\dom(f)$ implies $X_1 + q \in \dom(f)$. This proves that $\Sell(q)$ can successfully execute at $X$ (Observation~\ref{obs:generator-feasibility}), and we are left to show $Y(X, \Sell(q)) \geq q \cdot p$. Because  token 1 reserves at $X$ are lower than at $X'$, $\Sell(q)$ receives more tokens by executing at $X$  than $X'$ (Lemma~\ref{lemma:pricing}). Thus $\Sell(q)$ receives at least $q \cdot p$ tokens by executing at $X$, since $\Sell(q)$ receives at least $q \cdot p$ by executing at $X'$. This proves $\Sell(q, p)$ can successfully execute at $X$.
\end{proof}

\begin{corollary}\label{cor:execution}
Consider states $X, X' \in L_c(\phi)$, where $X_1 \leq X_1'$. $\Buy(q, p)$ receives a better execution at $X'$  than $X$. Moreover, $\Sell(q, p)$ receives a better execution at $X$  than $X'$.
\end{corollary}
\begin{proof}
If $\Buy(q, p)$ fails to execute at $X$, then $\Buy(q, p)$ receives a better execution at $X'$ (by definition). If $\Buy(q, p)$ can successfully execute at $X$, then $\Buy(q, p)$ also successfully executes at $X'$ (Lemma~\ref{lemma:execution-monotonicity}). From the Pricing Lemma (Lemma~\ref{lemma:pricing}), we have $Y(X', \Buy(q)) \leq Y(X, \Buy(q))$. This proves that $\Buy(q, p)$ receives a better execution at $X'$  than $X$.

If $\Sell(q, p)$ fails to execute at $X'$, then $\Sell(q, p)$ receives a better execution at $X$ (by definition). If $\Sell(q, p)$ can successfully execute at $X'$, then $\Sell(q, p)$  can also successfully execute at $X$ (Lemma~\ref{lemma:execution-monotonicity}). From the Pricing Lemma (Lemma~\ref{lemma:pricing}), we have $Y(X, \Sell(q)) \geq Y(X', \Sell(q))$. This proves  that $\Sell(q, p)$ receives a better execution at $X$  than $X'$.
\end{proof}

\begin{proof}[Proof of Theorem~\ref{thm:duality}]
First, consider the case where the token 1 reserves at $X$ are smaller than those at $X'$.  From Corollary~\ref{cor:execution}, any sell order receives a better execution at $X$  than $X'$. Second, consider the case where the token 1 reserves at $X$ are larger than at $X'$. From Corollary~\ref{cor:execution}, any buy order receives a better execution at $X$  than $X'$.
\end{proof}

\bibliographystyle{abbrvnat}
\bibliography{mybib}

\end{document}